%% file: courtade-kumar.tex
\documentclass[journal]{IEEEtran}

\usepackage{amsmath,amssymb,amsthm,mathtools}
\usepackage{graphicx}
\usepackage{booktabs,array}
\usepackage{cite}
\usepackage{microtype}
\usepackage[hidelinks]{hyperref}
\usepackage{url}
\hypersetup{
  pdftitle={The Most Informative Bit and Beyond: A Proof of the Courtade--Kumar Conjecture and Multibit Extensions},
  pdfauthor={Hessam Mahdavifar and Ahmad Beirami}
}

\newtheorem{theorem}{Theorem}
\newtheorem{lemma}[theorem]{Lemma}
\newtheorem{appendixlemma}[theorem]{Lemma}
\newtheorem{proposition}[theorem]{Proposition}

\newtheorem{conjecture}[theorem]{Conjecture}
\theoremstyle{definition}

\theoremstyle{remark}
\newtheorem{remark}[theorem]{Remark}

\newcommand{\E}{\mathbb E}
\newcommand{\Ent}{\operatorname{Ent}}
\newcommand{\Inf}{\operatorname{Inf}}
\newcommand{\atanh}{\operatorname{atanh}}
\newcommand{\sech}{\operatorname{sech}}

\newcommand{\diff}{\,\mathrm d}
\newcommand{\Lapl}{\mathcal L}
\newcommand{\degree}{d}
\newcommand{\logodds}{\ell_\rho}
\newcommand{\weight}{\omega}
\newcommand{\Klsi}{K_{\mathrm{LS}}}
\newcommand{\cutoff}{v_{\mathrm{cut}}}
\newcommand{\edgewidth}{d_{\mathrm{edge}}}
\newcommand{\corr}{\operatorname{corr}}
\newcommand{\unitcost}{\operatorname{yield}}

\begin{document}

\title{The Most Informative Bit and Beyond: \\A Proof of the Courtade--Kumar Conjecture \\and Multibit Extensions}

\author{Hessam~Mahdavifar and Ahmad~Beirami
\thanks{H. Mahdavifar is with the Department of Electrical and Computer
Engineering, Northeastern University, Boston, MA 02115 USA
(e-mail: h.mahdavifar@northeastern.edu).}
\thanks{A. Beirami is with Fidian, San Francisco,
CA, USA (e-mail: ab@fidian.ai).}%
}

\maketitle

\begin{abstract}
Let \(X=(X_1,\ldots,X_n)\) be uniform on the Boolean cube, and let \(Y\)
be obtained by passing its coordinates independently through a binary
symmetric channel with crossover probability \(p\). A natural quantization
question asks how much information about \(Y\) can be retained by a
function \(Q(X)\) constrained to \(k\) output bits. For \(k\leq n\),
reporting \(k\) input coordinates retains \(k(1-h_2(p))\) bits of
information, where \(h_2\) denotes the binary entropy function. This
makes coordinate projections the natural benchmark.

The case \(k=1\) is the Courtade--Kumar conjecture, which posits that every
Boolean function \(f\) satisfies
\[
    I(f(X);Y)\leq 1-h_2(p),
\]
with equality attained by coordinate functions. We prove this conjecture
for every dimension and crossover probability, without any balance
assumption on \(f\). Our proof views the gap relative to the coordinate
benchmark as a quantity that evolves with the channel correlation. After
a monotone rearrangement, an entropy-flow identity relates its derivative
to the edge boundary of the corresponding Boolean decision set. Combining
a sharp local entropy comparison with Fourier analysis on the Boolean cube
and logarithmic Sobolev estimates for pivotal sets yields a global
differential inequality: if the gap is positive at any correlation level,
then it must grow as the noise decreases. This contradicts the noiseless
endpoint, where the gap is nonpositive.

For \(k\geq8\), however, a distinctly coding-theoretic behavior emerges.
We construct quantizers from shortened Hamming codes whose covering
structure beats the coordinate-projection benchmark for \(k=8\), and a
direct-sum argument extends this failure to every \(k\geq8\). Thus, unlike
in the single-bit case, structured coding can retain more information than
simple coordinate selection. Finally, searches within structured code
families for \(2\leq k\leq7\), including exhaustive checks in selected
families, and Monte Carlo searches over seven-bit random codebooks find
no violation below eight bits. These findings lead us to conjecture that
coordinate projections remain optimal for \(2\leq k\leq7\).
\end{abstract}


\input{introduction}
\input{mechanism}
\input{multibit-extension}
\input{conclusion}
\input{ai-disclosure}

\appendices
\input{scalar}
\input{centered-support}
\input{pivotal-channel-checks}
\input{low}
\input{endpoints}

\input{verification}

\input{references}
\end{document}

%% file: introduction.tex
\section{Introduction}
\label{sec:introduction}

Let \(X=(X_1,\ldots,X_n)\) be a sequence of independent uniform binary
random variables, and let \(Y\) be obtained by passing every coordinate
through a binary symmetric channel (BSC) with crossover probability
\(p\in[0,1/2]\). A Boolean summary \(f(X)\) retains one bit from the
input. Courtade and Kumar conjectured that the mutual information between
this summary and the complete noisy observation is maximized by retaining
one of the original coordinates~\cite{CK}. Equivalently,
\begin{equation}
 I(f(X);Y)\leq 1-h_2(p),
 \label{eq:intro-ck}
\end{equation}
where \(h_2\) is the binary entropy function. Equality is attained by any
\textit{dictator} defined as \(f(X)=X_i\) for some \(i \in \{1,2,\dots,n\}\).

The problem is different from maximizing correlation or noise
stability. The observation \(Y\) is an \(n\)-dimensional random vector,
and information carried jointly by its coordinates need not decompose
into coordinatewise contributions. For example, the parity of two noisy
bits can contain information that is absent from either noisy bit alone.
Consequently, results for two Boolean outputs~\cite{AGKN,PPM} or for sums
of coordinatewise mutual informations~\cite{JW} do not directly imply
\eqref{eq:intro-ck}. The output \(f(X)\) may also be biased, which removes
several symmetries available in the balanced case.

This paper establishes \eqref{eq:intro-ck} in full generality. The proof
uses the correlation parameter \(\rho=1-2p\). After monotone sorting, the
mutual information becomes an entropy functional of the posterior mean
\(T_\rho f\). Its derivative admits an edge decomposition. The central
step is to strengthen the entropy-constrained edge envelope by retaining
the squared angular displacement lost in a Cauchy--Schwarz comparison.
The pivotal sets of the same Boolean function supply the Fourier and
logarithmic-Sobolev control needed to pay for this correction. The final
argument is dynamical: wherever a function has more information than a
dictator, its advantage must increase with \(\rho\). Such an advantage
cannot return to its nonpositive value at \(\rho=1\).

A Lean~4 formalization of our one-bit proof, together with its
definitions, pinned dependencies, and verification instructions, is
available in the project's \href{https://github.com/hessammahdavifar/most-informative-bit/tree/main/lean}{\texttt{lean} directory}.\footnote{\url{https://github.com/hessammahdavifar/most-informative-bit/tree/main/lean}}
Its exported statement uses finite-distribution mutual information in
bits and covers every crossover probability \(p\in[0,1]\).

We also study the natural vector-valued extension in which \(f(X)\)
contains \(k\) bits. The scalar theorem does not tensorize through the
chain rule because conditioning on preceding output bits destroys the
uniform input law. In fact, the coordinate benchmark is false for
sufficiently large fixed \(k\). We give explicit shortened-Hamming
counterexamples for \(k=8,9,10\), extend them to every \(k\geq8\), and
formulate a threshold conjecture for \(2\leq k\leq7\).

\subsection{Related Work}

Courtade and Kumar introduced the full-vector problem and proved several
structural reductions, including the monotone sorting operation used
below~\cite{CK}. The conjecture stimulated a sequence of Fourier,
functional-inequality, and information-geometric approaches. The
Friedgut--Kalai--Naor theorem identifies the rigidity of Boolean functions
whose Fourier mass is concentrated on the first two levels~\cite{FKN}.
Samorodnitsky established the conjecture in a dimension-free high-noise
regime and developed strengthened inequalities for the entropy of noisy
Boolean functions~\cite{Sam,SamEntropy}. Ordentlich, Shayevitz, and
Weinstein obtained sharper balanced high-noise bounds using Fourier
analysis and hypercontractivity~\cite{OSW}, while Yang and Wesel gave a
complementary calculus argument in a dimension-dependent neighborhood of
complete noise~\cite{YW}.

Yu developed a general \(\Phi\)-stability framework and subsequently
proved broad local-optimality and balanced finite-range results using
quantitative Fourier estimates~\cite{YuPhi,Yu,YuSharp}. Javanmard and
Woodruff proved the dictator bound for the sum of the individual-coordinate
mutual informations and sharpened the high-noise analysis~\cite{JW}; the
related Boolean-hull and moment relaxations are discussed in~\cite{Gemini}.
Li and M\'edard connected the problem to posterior moments and
noninteractive correlation distillation~\cite{LM}, and Barnes and
\"Ozg\"ur showed that the balanced conjecture is essentially equivalent
to a symmetrized Li--M\'edard formulation~\cite{BO}.

A second line of work seeks stronger entropy or isoperimetric statements.
Anantharam, Bogdanov, Chakrabarti, Jayram, and Nair proposed a Hellinger
inequality implying the conjecture~\cite{Hellinger}. Chen and Nair placed
this proposal in an ordered family of entropy inequalities and related
the low-noise limits to sharp cube isoperimetry~\cite{CN}. The critical
square-root isoperimetric inequality was recently established by Durcik,
Ivanisvili, Roos, and Xie~\cite{DIRX}, confirming the balanced low-noise
limit predicted by the Hellinger conjecture. This limiting result does
not by itself establish the finite-noise Shannon inequality. Other approaches use principal inertia components~\cite{PIC}, adjacency
events~\cite{HOS}, Gaussian analogues~\cite{KOW}, and stochastic
interpolation~\cite{EMR}.

The closest antecedent to the present proof is the differential-equation
framework of Chen, Gohari, and Nair~\cite{CGN}, introduced in 2025.
Their entropy-flow identity
and entropy-constrained two-point extremal formula are used here. 

Three concurrent works also establish the scalar conjecture in full
generality~\cite{CGJLMNW,KyTran,KramerSaglam}. Chen, Gohari, Javanmard, Lin, Mirrokni,
Nair, and Woodruff close the differential-equation program through an
unrestricted four-moment Bellman inequality for a hybrid potential,
using analytic reductions together with certified interval arithmetic.
Ky and Tran instead develop an entropy-production and spectral framework.
Their method is substantially closer to ours: both arguments use monotone
compression, edgewise entropy-production estimates, Fourier control of
functions with diffuse coordinate dependence, separate channel regimes,
and a differential contradiction at the noiseless endpoint. The
quantitative implementations differ. Our proof centers on an angular
edge correction and logarithmic-Sobolev estimates for pivotal sets,
whereas Ky and Tran use a profile-clock integration and selected-coordinate
spectral estimates. In another concurrent work, Kramer and
Saglam~\cite{KramerSaglam} give a substantially shorter proof based on
induction on the dimension, applied to a strengthened inequality that
retains the output mean. We have independently verified their proof in
Lean. Finally, neither of~\cite{CGJLMNW,KyTran,KramerSaglam} considers the \(k\)-bit extension
studied in the present paper.

As mentioned above, the present paper develops the vector-valued extension in
which the quantizer reports \(k\) bits. For this problem, Chandar and
Tchamkerten characterized the positive-rate asymptotics and
gave Hamming- and Golay-code counterexamples with \(k\geq11\)~\cite{CT}.
Huleihel and Ordentlich proved coordinate optimality at the complementary
endpoint \(k=n-1\)~\cite{HO}. The finite-block behavior observed in our
code search is also related qualitatively to recent results showing that
structured Reed--Solomon-based codebooks can outperform random codebooks
under average Hamming covering radius~\cite{RMcover}; the objective in
that work is different from BSC mutual information.

\subsection{Contributions and Organization}

Our first main contribution is Theorem~\ref{thm:CK}, which proves the
Courtade--Kumar conjecture in every dimension and for every crossover
probability, without any balance assumption on the Boolean function. The
proof follows a single dynamical strategy. In natural units, set
$
 I_f(\rho)\coloneqq I(f(X);Y_\rho),
 \phi(\rho)\coloneqq I(X_1;Y_\rho),
 \Delta_f(\rho)\coloneqq I_f(\rho)-\phi(\rho),
$.
We show that a positive advantage over a coordinate projection would have
to grow strictly as the channel becomes less noisy. Section~\ref{sec:mechanism}
introduces the posterior representation of \(I_f\), the required Fourier
notation, and the monotone reduction. In particular,
Lemma~\ref{lem:sorting} shows that coordinatewise sorting preserves any
advantage, allowing the remainder of the proof to exploit the
edge structure of a monotone Boolean decision set. An entropy-flow
identity expresses \(\rho I_f'(\rho)\) as an action functional of the
posterior mean \(T_\rho f\). Section~\ref{sec:edges} then decomposes this
action over edges of the discrete cube. The key estimate,
Lemma~\ref{lem:edges}, strengthens the usual two-point entropy envelope by
retaining a squared angular correction that is exact for dictators. The
rest of the proof shows that the Boolean and spectral structure of \(f\)
always provides enough reserve to pay for this correction.

Writing \(\ell=\operatorname{arctanh}\rho\), we divide the channel into
the small-correlation range \(0\leq\ell\leq31/20\), the intermediate
range \(31/20\leq\ell\leq7/2\), and the near-noiseless tail
\(\ell\geq7/2\). Equivalently, the two junctions are
\(\rho=\tanh(31/20)\) and \(\rho=\tanh(7/2)\). These cutoffs are technical
junctions at which the corresponding estimates meet, rather than
asserted phase transitions in the problem. Proposition~\ref{prop:low}
settles the small-correlation range directly, at every output bias, by
combining a Fourier energy gap with scalar entropy bounds. In the
intermediate range, the pivotal sets of \(f\) supply the needed reserve:
logarithmic Sobolev inequalities control the noisy entropies of their
indicators, and a log-sum argument combines the coordinatewise costs with
the global edge budget. This yields
Proposition~\ref{prop:middle-growth}, according to which
\(\Delta_f(\rho)>0\) implies \(\Delta_f'(\rho)>0\) throughout that range.
The bounded angular reserve is insufficient by itself near the noiseless
endpoint, so the tail argument supplements the same edgewise inequality
with degree-weighted Fourier information. Lemma~\ref{lem:deficit}
identifies the relevant coordinate deficits,
Lemma~\ref{lem:baseline} converts the information advantage into an
entropy-funded spectral baseline, and Lemma~\ref{lem:budget} turns the
resulting budget comparison into strict information growth.
Proposition~\ref{prop:local} separately rules out a counterexample having
a sufficiently influential coordinate. Section~\ref{sec:coverage}
assembles these complementary estimates in
Lemma~\ref{lem:coverage}, which gives the global implication
\[
 \Delta_f(\rho)>0\quad\Longrightarrow\quad\Delta_f'(\rho)>0.
\]
Since \(\Delta_f(1)=H(f(X))-\log 2\leq0\), a positive advantage cannot
persist to the noiseless endpoint; this contradiction completes the proof
of Theorem~\ref{thm:CK}.

Our second contribution, developed in
Section~\ref{sec:multibit-extension}, concerns the natural extension in
which the quantizer reports \(k\) bits and the coordinate benchmark is
\(k(1-h_2(p))\). Proposition~\ref{prop:affine-multibit} shows that affine
linear maps cannot violate this benchmark. Nonlinear nearest-codeword
quantizers built from linear codebooks behave differently: an exact
high-noise expansion, expressed through the coset-leader distribution,
produces shortened-Hamming counterexamples for \(k=8,9,10\); Proposition~\ref{prop:k8-k10-counterexamples}
extends the failure to every \(k\geq8\) by a direct-sum construction. Below eight bits, structured-code searches and seven-bit random-codebook
experiments find no violation. The seven-bit examples also show an
advantage of the tested structured codebooks over the sampled random
codebooks. These findings motivate
Conjecture~\ref{conj:seven-bit-threshold}, which predicts that coordinate
projections remain optimal for \(2\leq k\leq7\). The appendices establish
the scalar and Fourier bounds, endpoint reductions, and channel
comparisons needed for the proof.

%% file: mechanism.tex
\section{Preliminaries}
\label{sec:mechanism}

We fix the entropy and Fourier normalizations used in the one-bit proof.
The posterior representation turns mutual information into a function of
the channel correlation, and monotone reduction permits the edge and
pivotal-set estimates that follow.

\subsection{Notation}

Let \(X\) be uniform on \(\{-1,1\}^n\), and let \(Y\) be its
observation through a BSC with crossover probability \((1-\rho)/2\),
where \(\rho\in[0,1]\). Expectations are uniform on the appropriate
cube unless otherwise specified. Throughout the one-bit proof and its
appendices, logarithms are natural and information is measured in nats.
For \(r\in[-1,1]\), define
\begin{align*}
 \phi(r)&=\frac{1+r}{2}\log(1+r)
          +\frac{1-r}{2}\log(1-r),\\
 L&=\log2,\qquad h(r)=L-\phi(r).
\end{align*}
Thus \(h(r)\) is binary entropy in the mean parametrization.
We use the convention \(0\log0=0\).
Write \(g(r)=\phi'(r)=\atanh r\) for \(|r|<1\), and
\(\logodds=\atanh\rho\) for \(0<\rho<1\).

For \(f:\{-1,1\}^n\to\{-1,1\}\), let
\[
 m=\E f,\qquad u(y)=T_\rho f(y)=\E[f(X)\mid Y=y].
\]
Then
\begin{equation}
 \begin{aligned}
 I_f(\rho)&:=I(f(X);Y)\\
 &=h(m)-\E h(u)=\E\phi(u)-\phi(m).
 \end{aligned}
 \label{eq:posterior-information}
\end{equation}
A signed dictator \(f(x)=\pm x_i\) attains \(I_f(\rho)=\phi(\rho)\).
Write \(\Delta_f(\rho)=I_f(\rho)-\phi(\rho)\) for the information gap.
Constant functions have zero mutual information; in the one-bit argument, assume
\(f\) is nonconstant, so \(|u(y)|<1\) for \(0<\rho<1\).

\subsection{Fourier Analysis and Entropy Flow}

We use the standard Boolean Fourier conventions~\cite{ODonnell}.
For \(S\subseteq[n]\), let \(\chi_S(x)=\prod_{i\in S}x_i\) and
\(\widehat v(S)=\E[v\chi_S]\). For any real-valued cube function \(v\),
\begin{align*}
 v&=\sum_{S\subseteq[n]}\widehat v(S)\chi_S,
 &\E v^2&=\sum_S\widehat v(S)^2,\\
 T_\rho v&=\sum_S\rho^{|S|}\widehat v(S)\chi_S.
\end{align*}
The cube Laplacian is normalized by
\begin{align*}
 \Lapl\chi_S&=|S|\chi_S,\\
 \Lapl v(x)&=\frac12\sum_{i=1}^n
       \bigl(v(x)-v(x^{\oplus i})\bigr),
\end{align*}
where \(x^{\oplus i}\) denotes \(x\) with coordinate \(i\) flipped.
For fixed \(f\) and \(0<\rho<1\),
\(\rho\,\partial_\rho u=\Lapl u\), and the entropy-flow identity
\cite[Lemma~1]{CGN} gives
\begin{equation}
 A(u):=\E[g(u)\Lapl u]=\rho I_f'(\rho).
 \label{eq:action}
\end{equation}

\subsection{Monotone Reduction}

We use the following consequence of coordinate compression
\cite[Lemma~2 and Remark~2]{CK}.
\begin{lemma}[Monotone reduction]
\label{lem:sorting}
For every Boolean \(f\), there is a coordinatewise nondecreasing
Boolean \(\widetilde f\) in the same dimension such that
\(\E\widetilde f=\E f\) and
\(I_{\widetilde f}(\rho)\ge I_f(\rho)\) for all \(\rho\in[0,1]\).
\end{lemma}
It therefore suffices to consider a fixed nonconstant monotone
function throughout the one-bit proof.

\section{Edge Geometry and the Pivotal Decomposition}
\label{sec:edges}

Fix a nonconstant monotone \(f\), let \(0<\rho<1\), and write
\(u=T_\rho f\). We bound the information growth \(A(u)\) using
the geometry of the pivotal sets. The key estimate combines the
Dirichlet energy of \(\Theta=\arcsin u\) with an entropy-dependent
correction. The energy term will be controlled by Fourier analysis.

\begin{figure*}[t]
\centering
\includegraphics[width=.76\textwidth]{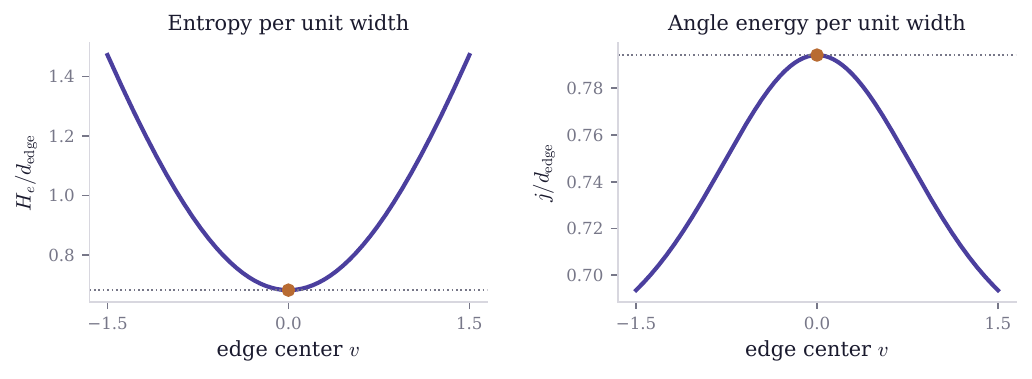}
\caption{Centered-edge comparison for \(p=\tanh(v+k)\),
\(q=\tanh(v-k)\), and \(k=0.8\). Here
\(\edgewidth=|p-q|/2\), \(H_e=[h(p)+h(q)]/2\), and
\(j=(\arcsin p-\arcsin q)^2/4\).
Centering at \(v=0\) minimizes \(H_e/\edgewidth\) and maximizes
\(j/\edgewidth\), yielding the scalar bounds used in
Lemma~\ref{lem:edges}.}
\label{fig:centered-edge}
\end{figure*}

Let \(f_{i,\pm}\) denote the sections of \(f\) with coordinate \(i\)
fixed to \(\pm1\). Define the pivotal indicator and its mean by
\begin{align*}
 q_i^f&=\frac{f_{i,+}-f_{i,-}}2\in\{0,1\},\\
 a_i&=\E q_i^f=\widehat f(\{i\})=\Inf_i(f).
\end{align*}
The indicators \(q_i^f\) will also enter the logarithmic Sobolev
estimates below. All coordinate sums in this section omit indices
with \(a_i=0\).

For \(v>0\), define
\[
 F(v)=\frac{h(\tanh v)}{\tanh v},\qquad
 c(v)=\frac{(\arcsin(\tanh v))^2}{\tanh v}.
\]
The function \(F\) is strictly decreasing from infinity to zero.
A common entropy budget determines the comparison heights in the
following two bounds. Figure~\ref{fig:centered-edge} illustrates the
underlying centered-edge comparison.

\begin{lemma}[The edge budget]
\label{lem:edges}
Let \(H_{\rm budget}\ge\E h(u)>0\). For each \(i\) with \(a_i>0\),
let \(v_i>0\) be the unique solution of
\[
 F(v_i)=\frac{H_{\rm budget}}{\rho a_i}.
\]
Then
\begin{align}
 A(u)&\ge\rho\sum_i a_i v_i,
 \label{eq:edge-budget}\\
 A(u)&\ge\E[\Theta\Lapl\Theta]
       +\rho\sum_i a_i\bigl(v_i-c(v_i)\bigr).
 \label{eq:correction}
\end{align}
\end{lemma}

The first inequality is the entropy-constrained edge bound
\cite[Theorem~4 and Appendix~C]{CGN}. Both bounds follow by averaging
convex scalar edge inequalities; the full proof is given in
Appendix~\ref{app:scalar}.

Whenever \(\E h(u)\le h(\rho)\), the choice
\(H_{\rm budget}=h(\rho)\) defines heights \(\ell_i\) satisfying
\begin{equation}
 F(\ell_i)=\frac{F(\logodds)}{a_i},\qquad
 0<\ell_i\le\logodds.
 \label{eq:heights}
\end{equation}
The first bound gives
\begin{equation}
 W:=\sum_i a_i\ell_i\le\frac{A(u)}\rho.
 \label{eq:W}
\end{equation}
The second bound will be combined with a spectral estimate for
\(\E[\Theta\Lapl\Theta]\), retaining the correction term.
The general budget formulation also permits the sharper entropy
bound used in the intermediate-correlation argument.

\section{Information Growth Across Channel Regimes}
\label{sec:pivotal-middle}

This section supplies the differential step of the proof: a positive
information gap forces strictly positive growth. In the intermediate
range, pivotal entropy controls the angular energy in the edge budget.
Near the noiseless endpoint, coordinatewise spectral deficit bounds
replace that estimate; both arguments retain the same angular comparison.
The small-correlation bound and the final contradiction are assembled in
Section~\ref{sec:coverage}.

\subsection{Intermediate-Correlation Growth}

Write \(r=\rho\), \(\ell=\atanh r\), and
\(\Delta_f(r)=I_f(r)-\phi(r)\). We combine an angular-energy bound
with logarithmic Sobolev control of pivotal sets, then close the edge
budget using the log-sum inequality.

\begin{proposition}[Intermediate-correlation growth]
\label{prop:middle-growth}
For every monotone Boolean \(f\) and \(31/20\le\ell\le7/2\),
\[
 \Delta_f(r)>0\quad\Longrightarrow\quad I_f'(r)>\phi'(r).
\]
\end{proposition}

\subsubsection{Angular-energy bound}

Set \(J_{\rm ang}=\E[\Theta\Lapl\Theta]\), where \(\Theta=\arcsin u\),
and define
\begin{align*}
 \theta&=\arcsin r,\qquad \lambda=\theta/r,\\
 \alpha&=\frac{\theta+\sinh\ell}{\ell},\qquad
 \gamma=\lambda\alpha,\\
 \eta&=\frac{\theta^2\gamma}{\gamma+1}.
\end{align*}
For \(\ell\ge31/20\) and \(x\in[-1,1]\),
Appendix~\ref{app:centered-support} gives
\begin{equation}
 \begin{split}
 \frac{\alpha}{2}(\arcsin x-\lambda x)^2
 &\le x\arcsin x-r\theta\\
 &\quad-\alpha[\phi(x)-\phi(r)].
 \end{split}
 \label{eq:centered-support}
\end{equation}
Averaging at \(x=u\), multiplying by \(2\lambda\), and completing
the square in \(\Lapl+\gamma\) yields
\begin{equation}
 \begin{split}
 J_{\rm ang}\ge{}&2\theta^2+2\gamma[\phi(m)+\Delta_f(r)]\\
 &+\lambda^2\E\left[u\left\{\gamma-
       (\gamma+1)^2(\Lapl+\gamma)^{-1}\right\}u\right].
 \end{split}
 \label{eq:centered-resolvent}
\end{equation}

\subsubsection{Pivotal entropy}

For \(a_i>0\) and \(0\le t\le r^2\), let
\[
 w_i(t)=T_{\sqrt t}q_i^f,\qquad s_i(t)=\E w_i(t)^2,
\]
with noise acting on the other coordinates. All sums below run over
\(a_i>0\). Define
\[
 z_i=\frac{\gamma+1}{\gamma r^2}
       \int_0^{r^2}[1-(t/r^2)^\gamma]s_i(t)\,dt,
 \qquad Z=\sum_i z_i,
\]
and put \(V_0=1-m^2\), \(U=\operatorname{Var}(u)/r^2\).
The integral is a probability average, so \(a_i^2\le z_i\le a_i\).
Fourier expansion gives
\[
 Z=(\gamma+1)\sum_{k\ge1}
       \frac{r^{2(k-1)}}{\gamma+k}
       \sum_{|S|=k}\widehat f(S)^2,
 \qquad 0<Z\le U\le V_0.
\]

Gross's logarithmic Sobolev inequality~\cite{Gross,ODonnell},
together with Jensen's inequality, gives
\[
 \E[w_i\Lapl w_i]\ge
       \frac12 s_i\log\frac{s_i}{a_i^2}.
\]
Averaging with the kernel defining \(z_i\), applying convexity of
\(s\log(s/a_i^2)\), and expanding in Fourier degrees yields
\begin{equation}
 (\gamma+1)(U-Z)\ge\frac12 E_{\rm piv},\qquad
 E_{\rm piv}:=\sum_i z_i\log\frac{z_i}{a_i^2}.
 \label{eq:pivotal-entropy}
\end{equation}
With
\[
 M(m)=2\gamma\phi(m)
       -[\lambda^2(2-r^2)+\lambda/\alpha]m^2,
\]
combining \eqref{eq:pivotal-entropy} with the Fourier expansion of
\eqref{eq:centered-resolvent} gives
\begin{equation}
 \begin{split}
 J_{\rm ang}\ge{}&\theta^2(1+V_0-Z)+\frac\eta2 E_{\rm piv}\\
 &+M(m)+2\gamma\Delta_f(r).
 \end{split}
 \label{eq:pivotal-baseline}
\end{equation}

\subsubsection{Closing the edge budget}

Assume \(31/20\le\ell\le7/2\), \(\Delta=\Delta_f(r)>0\), and
\(A(u)\le r\ell\). Monotonicity gives \(a_i\le1-|m|\le V_0\).
Since \(\phi(m)\ge m^2/2\) and \(h(r)<1/2\) in this range,
\[
 \E h(u)=h(r)-\phi(m)-\Delta\le V_0h(r).
\]
Apply Lemma~\ref{lem:edges} with \(H_{\rm budget}=V_0h(r)\).
Its heights satisfy
\[
 F(v_i)=\frac{V_0F(\ell)}{a_i},\qquad 0<v_i\le\ell.
\]
Writing \(H(v)=v-c(v)\ge0\), we obtain
\begin{equation}
 \sum_i a_iv_i\le\ell,\qquad
 A(u)\ge J_{\rm ang}+r\sum_i a_iH(v_i).
 \label{eq:adjusted-edge}
\end{equation}

Put \(\beta=2r/\eta\). Lemma~\ref{lem:middle-channel-checks} and
\(2\phi(m)\ge m^2\) imply
\begin{align}
 \log\frac{vF(v)}{\ell F(\ell)}
 &\ge\beta[H(\ell)-H(v)],\quad 0<v\le\ell,
 \label{eq:middle-profile}\\
 \theta^2-rH(\ell)-\eta/2&\ge0,
 \label{eq:middle-energy}\\
 M(m)-rH(\ell)m^2&\ge0.
 \label{eq:middle-mean}
\end{align}
Set \(b_i=a_i^2e^{-\beta H(v_i)}\) and
\(B_{\rm piv}=\sum_i b_i\). By \eqref{eq:middle-profile} and
\eqref{eq:adjusted-edge},
\[
 B_{\rm piv}\le
 \frac{V_0}{\ell}e^{-\beta H(\ell)}\sum_i a_iv_i
 \le V_0e^{-\beta H(\ell)}.
\]
Using \(z_i\le a_i\) and \(\eta\beta/2=r\), the log-sum inequality
therefore gives
\begin{align*}
 \frac\eta2 E_{\rm piv}+r\sum_i a_iH(v_i)
 &\ge\frac\eta2\sum_i z_i\log\frac{z_i}{b_i}\\
 &\ge\frac\eta2 Z\log\frac Z{B_{\rm piv}}\\
 &\ge\frac\eta2 Z\log\frac Z{V_0}+rZH(\ell).
\end{align*}
Combining this with \eqref{eq:pivotal-baseline} and
\eqref{eq:adjusted-edge}, then using
\(r\ell=\theta^2+rH(\ell)\) and
\(Z\log(Z/V_0)\ge Z-V_0\), yields
\[
 \begin{aligned}
 A(u)-r\ell\ge{}&
 [\theta^2-rH(\ell)-\eta/2](V_0-Z)\\
 &+M(m)-rH(\ell)m^2+2\gamma\Delta>0,
 \end{aligned}
\]
where the strict inequality follows from
\eqref{eq:middle-energy}--\eqref{eq:middle-mean} and \(\Delta>0\).
This contradicts \(A(u)\le r\ell\) and proves the proposition by
\eqref{eq:action}.

\subsection{High-Correlation Spectral Bounds}
\label{sec:coordinate}

For \(\ell=\atanh\rho\ge7/2\), we estimate the spectral term in
\eqref{eq:centered-resolvent} using a reference Fourier degree.
Retain \(\gamma=\lambda\alpha\) and let \(\degree\ge2\) be real;
the tail argument uses \(\degree=\ell/2+5/4\). Define
\begin{align*}
 \weight(s)&=\frac{\gamma s}{\gamma+s},\\
 \Klsi&=\frac{\gamma^2e^{2\degree-3}}
              {2(\gamma+\degree)(\gamma+1)},\\
 d_j&=[\weight(\degree)-\weight(j)]\rho^{2j},
       \qquad j=1,2.
\end{align*}

\begin{lemma}[Coordinate deficit bounds]\label{lem:deficit}
Let \(u=T_\rho f\), where \(f\) is monotone Boolean and \(0<\rho<1\).
Let \(\gamma>0\), \(\degree\ge2\), and set
\begin{align*}
 D_i&=\sum_{S\ni i}
       \frac{\weight(\degree)-\weight(|S|)}{|S|}\widehat u(S)^2,\\
 D(a)&=\min\left\{\Klsi\rho^2a^2,\,
       \frac{d_2}{2}a+\left(d_1-\frac{d_2}{2}\right)a^2\right\}.
\end{align*}
Then \(D_i\le D(a_i)\), and
\begin{equation}
 \sum_{S\ne\varnothing}\weight(|S|)\widehat u(S)^2
 \ge \weight(\degree)\operatorname{Var}(u)-\sum_iD(a_i).
 \label{eq:spectral}
\end{equation}
\end{lemma}

\begin{proof}
The case \(a_i=0\) is immediate. For \(w\ge0\), write
\(b=\E w>0\) and \(s=\E w^2\). The cube logarithmic Sobolev
inequality~\cite{Gross,ODonnell} and Jensen's inequality give
\[
 \begin{aligned}
 \E[w(\Lapl+1-\degree)w]
 &\ge \frac{s}{2}\log\frac{s}{b^2}+(1-\degree)s\\
 &\ge -\frac12e^{2\degree-3}b^2,
 \end{aligned}
\]
where the last step minimizes over \(s>0\).

Set \(p_i=(u_{i,+}-u_{i,-})/2\) and \(w_t=T_{\sqrt t}p_i\),
with noise and \(\Lapl_{-i}\) acting on the remaining coordinates.
Then \(w_t\ge0\) and \(\E w_t=\rho a_i\).
Fourier expansion gives
\[
 D_i=\frac{\gamma}{\gamma+\degree}
 \int_0^1(1-t^\gamma)
 \E[w_t(\degree-1-\Lapl_{-i})w_t]\,\diff t,
\]
using \(\int_0^1(1-t^\gamma)t^{s-1}\,\diff t
=\gamma/[s(\gamma+s)]\).
The preceding estimate and
\(\int_0^1(1-t^\gamma)\,\diff t=\gamma/(\gamma+1)\) yield
\(D_i\le\Klsi\rho^2a_i^2\).

For the second bound, Parseval applied to the pivotal indicator
\(q_i^f\) gives
\(\sum_{S\ni i}\widehat f(S)^2=a_i\), with singleton contribution
\(a_i^2\). For \(s\ge2\),
\([\weight(\degree)-\weight(s)]\rho^{2s}\le d_2\):
both nonnegative factors decrease on \(2\le s\le\degree\),
and the product is nonpositive for \(s>\degree\). Hence
\[
 D_i\le d_1a_i^2+\frac{d_2}{2}(a_i-a_i^2),
\]
which proves the second branch of \(D(a_i)\).
Finally,
\(\sum_iD_i=\sum_{S\ne\varnothing}
[\weight(\degree)-\weight(|S|)]\widehat u(S)^2\);
summing the coordinate bounds proves \eqref{eq:spectral}.
\end{proof}

\subsection{Spectral Baseline}
\label{sec:baseline}

We now combine the angular comparison with the coordinate deficit bounds.
The information gap supplies a lower bound on the noisy second moment,
leaving only scalar channel conditions to control the output bias.

Set
\begin{align*}
 p&=(\lambda+1/\alpha)^2,\qquad
 b=\lambda^2(2+1/\gamma),\\
 C_d&=p\weight(\degree)-b,\\
 B&=2\theta^2+C_d\frac{\phi(\rho)}L,\qquad
 k_\Delta=2\gamma+\frac{C_d}L.
\end{align*}

\begin{lemma}[Spectral baseline]\label{lem:baseline}
Let \(u=T_\rho f\) for monotone Boolean \(f\), with
\(0<\rho<1\), \(\ell=\atanh\rho\ge31/20\),
\(\gamma=\lambda\alpha\), and \(\degree\ge2\).
If
\begin{equation}
 C_d\ge0,\qquad
 \gamma-b+C_d\left(\frac1{2L}-1\right)\ge0,
 \label{eq:canonical-baseline-conditions}
\end{equation}
then
\begin{equation}
 \E[\Theta\Lapl\Theta]\ge B-p\sum_iD(a_i)
                         +k_\Delta\Delta_f(\rho).
 \label{eq:baseline}
\end{equation}
\end{lemma}

\begin{proof}
Since \(\gamma=\lambda\alpha\),
\[
 \lambda^2\left[\gamma-\frac{(\gamma+1)^2}{\gamma+s}\right]
 =p\weight(s)-b.
\]
Substitute this into \eqref{eq:centered-resolvent} and apply
Lemma~\ref{lem:deficit}. Using \(C_d\ge0\) and
\[
 \E\phi(u)=\phi(\rho)+\phi(m)+\Delta_f(\rho)\le L\E u^2
\]
gives
\[
 \begin{aligned}
 \E[\Theta\Lapl\Theta]\ge{}&
 B-p\sum_iD(a_i)+k_\Delta\Delta_f(\rho)\\
 &+k_\Delta\phi(m)-(C_d+b)m^2.
 \end{aligned}
\]
The last two terms have nonnegative sum by
\(\phi(m)\ge m^2/2\) and
\eqref{eq:canonical-baseline-conditions}, proving the claim.
\end{proof}

\subsection{Yield Comparison and Growth}
\label{sec:contradiction}

The two edge bounds now close the argument. We express each coordinate's
angular cost and spectral loss per unit of edge weight, then compare this
quantity with the common baseline.

For \(0<v\le\ell=\atanh\rho\), define
\[
 \unitcost(v)=\frac{\rho c(v)}v+\frac{pD(a)}{av},
 \qquad a=\frac{F(\ell)}{F(v)}.
\]

\begin{lemma}[Information growth]\label{lem:budget}
Under the hypotheses of Lemma~\ref{lem:baseline}, suppose
\(\Delta_f(\rho)\ge0\) and
\begin{equation}
 \frac B\ell\ge\rho,\qquad
 \frac B\ell\ge\unitcost(\ell_i)
 \quad\text{for all }i\text{ with }a_i>0.
 \label{eq:scalar-test}
\end{equation}
Then
\begin{equation}
 \Delta_f'(\rho)\ge
 \frac{k_\Delta\ell}{B}\,\Delta_f(\rho).
 \label{eq:final-budget}
\end{equation}
In particular, \(\Delta_f(\rho)>0\) implies \(\Delta_f'(\rho)>0\).
\end{lemma}

\begin{proof}
Since \(\Delta_f(\rho)\ge0\),
\(\E h(u)=h(\rho)-\phi(m)-\Delta_f(\rho)\le h(\rho)\).
Lemmas~\ref{lem:edges} and~\ref{lem:baseline}, together with
\eqref{eq:scalar-test}, therefore give
\begin{align*}
 A(u)&\ge B+k_\Delta\Delta_f(\rho)
       +\sum_i a_i\ell_i[\rho-\unitcost(\ell_i)]\\
 &\ge B+k_\Delta\Delta_f(\rho)
       +\left(\rho-\frac B\ell\right)W.
\end{align*}
Using \(W\le A(u)/\rho\) and \(\rho-B/\ell\le0\), we obtain
\[
 \frac{B}{\rho\ell}A(u)\ge B+k_\Delta\Delta_f(\rho).
\]
Now \(B/\ell\ge\rho>0\) and
\(\Delta_f'(\rho)=A(u)/\rho-\ell\) yield
\eqref{eq:final-budget}.
\end{proof}

\section{Proof of the Main Theorem}
\label{sec:coverage}

We combine the direct small-correlation and influential-coordinate bounds
with the growth estimates of Section~\ref{sec:pivotal-middle}. Together
they show that a positive gap would have positive derivative at its
interior maximum, completing the contradiction.

\begin{theorem}[Main Theorem]
\label{thm:CK}
For every \(n\ge1\), every Boolean
\(f:\{-1,1\}^n\to\{-1,1\}\), and every \(0\le\rho\le1\),
\begin{equation}
 I(f(X);Y)\le\phi(\rho).
 \label{eq:main-ck-sign}
\end{equation}
\end{theorem}

Dividing by \(\log2\) and setting \(\rho=1-2p\) gives
\eqref{eq:intro-ck} for \(0\le p\le1/2\). Complementing every bit of
\(Y\) preserves mutual information and replaces \(p\) by \(1-p\),
so the same bound holds for \(1/2<p\le1\).

\subsection{Direct Bounds}

Set \(\rho_*=\tanh(31/20)\).

\begin{proposition}[Small correlation]\label{prop:low}
For every Boolean \(f\) and \(0\le\rho\le\rho_*\),
\[
 I_f(\rho)\le\phi(\rho).
\]
\end{proposition}
\begin{proof}
See Appendix~\ref{app:low}.
\end{proof}

\begin{proposition}[Influential coordinate]\label{prop:local}
Let \(f\) be Boolean, \(0<\rho<1\), and
\(a=|\widehat f(\{i\})|\) for some coordinate \(i\).
If
\begin{equation}
 (1-\rho^2)h(\rho a)\ge(1-\rho^2a)h(\rho),
 \label{eq:local}
\end{equation}
then \(I_f(\rho)\le\phi(\rho)\).
\end{proposition}
\begin{proof}
Appendix~\ref{app:scalar} derives a mean-dependent form of Yu's local
comparison~\cite[Theorem~4.6 and Remark~4.7]{Yu}.
\end{proof}

\subsection{Global Growth Comparison}

\begin{lemma}[Growth of a positive gap]\label{lem:coverage}
For every monotone Boolean \(f\) and \(0<\rho<1\),
\[
 I_f(\rho)>\phi(\rho)
 \quad\Longrightarrow\quad I_f'(\rho)>\phi'(\rho).
\]
\end{lemma}

\begin{proof}
Assume \(I_f(\rho)>\phi(\rho)\) and write \(\ell=\atanh\rho\).
Proposition~\ref{prop:low} excludes \(\ell\le31/20\), and
Proposition~\ref{prop:middle-growth} covers
\(31/20\le\ell\le7/2\).

For \(\ell\ge7/2\), Proposition~\ref{prop:local} implies that
\eqref{eq:local} fails for every coordinate.
Appendix~\ref{app:endpoints} reduces the height comparison to
the half-height join and cutoff. With \(\gamma=\lambda\alpha\)
and \(\degree=\ell/2+5/4\), Appendix~\ref{app:verification}
establishes \eqref{eq:canonical-baseline-conditions} and
\[
 \frac B\ell\ge\max\{\rho,\unitcost(\ell_i)\}
 \qquad\text{for every }i\text{ with }a_i>0.
\]
Lemma~\ref{lem:budget} therefore gives \(\Delta_f'(\rho)>0\).
\end{proof}

\begin{proof}[Proof of Theorem~\ref{thm:CK}]
Constants are immediate. By Lemma~\ref{lem:sorting}, it suffices
to consider nonconstant monotone \(f\).
The function \(\Delta_f=I_f-\phi\) is continuous on \([0,1]\),
differentiable on \((0,1)\), and satisfies
\[
 \Delta_f(0)=0,\qquad \Delta_f(1)=-\phi(m)\le0.
\]
If it were positive somewhere, it would attain a positive
maximum at an interior point where \(\Delta_f'=0\),
contradicting Lemma~\ref{lem:coverage}.
\end{proof}

%% file: multibit-extension.tex
\section{Vector-Valued Quantization: Counterexamples from Eight Bits}
\label{sec:multibit-extension}

Having established the one-bit bound in Theorem~\ref{thm:CK}, we ask
whether a \(k\)-bit summary is optimized by retaining \(k\) input
coordinates. Failure for large \(k\) is known~\cite{CT}; here we give
explicit counterexamples beginning at eight bits and examine the
smaller-output regime.

Shortened Hamming quantizers violate
the coordinate benchmark for \(k=8,9,10\), and hence, by adding independent
coordinate outputs, for every \(k\geq 8\). The same construction misses the
benchmark at \(k=7\) by only \(1/64\) in the high-noise coefficient. A broad
collection of searches below eight bits also failed to produce a violation.
This leads us to conjecture that the coordinate projection is optimal for
\(2\leq k\leq7\).

Throughout this section, entropies and mutual informations are measured in
bits. Let
\begin{align*}
 X&\sim\operatorname{Unif}(\mathbb F_2^n),\qquad Y=X\oplus Z,\\
 Z_i&\stackrel{\mathrm{iid}}{\sim}\operatorname{Bernoulli}(p),
 \qquad 0\leq p\leq\frac12.
\end{align*}
with \(Z\) independent of \(X\), and write
\[
 h_2(p)=-p\log_2p-(1-p)\log_2(1-p),
 \qquad C(p)=1-h_2(p).
\]

\subsection{Problem Formulation and Known Results}

For \(n\geq k\), define
\begin{equation}
 M_k(n,p):=\max_{f:\mathbb F_2^n\to\mathbb F_2^k}I(f(X);Y).
 \label{eq:multibit-Mkn}
\end{equation}
The coordinate projection \(D_k(x)=(x_1,\ldots,x_k)\) attains
\(kC(p)\). The vector extension asks whether
\begin{equation}
 M_k(n,p)\leq kC(p)\qquad\text{for every }n\geq k.
 \label{eq:vector-CK}
\end{equation}

Besides \(k=1\), coordinate optimality is known for \(k=n-1\), including
binary-input memoryless output-symmetric channels~\cite{HO}.
Chandar and Tchamkerten characterized the positive-rate asymptotics and
gave counterexamples from the \([15,11,3]\) Hamming and \([23,12,7]\)
Golay codes, asking for constructions with \(k<11\)~\cite{CT}.

\begin{proposition}[Affine maps obey the coordinate bound]
\label{prop:affine-multibit}
If \(f(x)=Ax+b\) has rank \(r\leq k\), then
\[
 I(f(X);Y)\leq rC(p)\leq kC(p).
\]
\end{proposition}
\begin{proof}
Uniformity makes \(Z\) independent of \(Y\) and gives
\(I(f(X);Y)=r-H(AZ)\).
For \(r\) independent columns of \(A\), indexed by \(J\),
\[
 H(AZ)\geq H(AZ\mid Z_{J^c})=r h_2(p).
\]
\end{proof}

\subsection{Coset Quantizers and the High-Noise Expansion}
\label{subsec:coset-diagnostic}

Let \(H\in\mathbb F_2^{(n-k)\times n}\) have full row rank and
\(C=\ker H\) have dimension \(k\).
Choose a minimum-weight leader \(e_a\) for each syndrome, fixing ties,
and define the coset quantizer
\begin{equation}
 Q(x)=x\oplus e_{Hx}\in C.
 \label{eq:coset-quantizer}
\end{equation}
Identify \(C\) with \(\mathbb F_2^k\). For uniform \(X\),
\(E=e_{HX}\) is uniform over the leaders and independent of \(Q(X)\);
hence
\begin{equation}
 I(Q(X);Y)=n-H(E\oplus Z).
 \label{eq:coset-information}
\end{equation}

For any fixed quantizer \(Q:\mathbb F_2^n\to\mathbb F_2^k\),
write \(U=Q(X)\) and set
\begin{equation}
 B(Q)=\sum_{j=1}^n
 \mathbb E\!\left[\left(\mathbb E[(-1)^{X_j}\mid U]\right)^2\right].
 \label{eq:general-B}
\end{equation}
With \(\rho=1-2p\), as \(\rho\to0\) the high-noise expansion is
\begin{equation}
 I(U;Y)-kC(p)
 =\frac{\rho^2}{2\ln2}\bigl(B(Q)-k\bigr)+O(\rho^4).
 \label{eq:general-high-noise}
\end{equation}
To see this, for each output of positive probability put
\(m_j(u)=\mathbb E[(-1)^{X_j}\mid U=u]\). The density of
\(Y\mid U=u\) relative to the uniform measure is
\[
 1+\rho\sum_jm_j(u)(-1)^{y_j}+O(\rho^2).
\]
Taylor expansion of relative entropy and Walsh orthogonality give
the coefficient \(B(Q)/(2\ln2)\). The expansion is analytic and
even in \(\rho\): changing its sign complements every output
coordinate without changing mutual information. Subtracting the
expansion of \(kC(p)\) proves \eqref{eq:general-high-noise}.
For coset quantizers, translation within each fiber gives
\[
 B(Q)=B(E):=\sum_{j=1}^n\left(\mathbb E[(-1)^{E_j}]\right)^2.
\]
Thus \(B(E)>k\) gives a counterexample for all \(p<1/2\)
sufficiently close to \(1/2\).

\subsection{Shortened-Hamming Counterexamples}
\label{subsec:shortened-hamming}

Identify nonzero vectors in \(\mathbb F_2^4\) with \(1,\ldots,15\)
through their binary expansions. For \(s=1,\ldots,6\), let \(H_s\)
have columns \(s+1,\ldots,15\), so
\[
C_s=\ker H_s=[15-s,11-s,3],\quad n=15-s,\quad k=11-s.
\]
Indeed, columns \(8,9,10,12\) give rank four, while the columns
are distinct and nonzero and \(7\oplus8\oplus15=0\).
Writing \(\mathbf e_a\) for the unit vector indexed by column \(a\),
choose the minimum-weight leaders
\begin{equation}
 e_a=
 \begin{cases}
 0, & a=0,\\
 \mathbf e_a, & a\in\{s+1,\ldots,15\},\\
 \mathbf e_8\oplus\mathbf e_{8\oplus a}, & a\in\{1,\ldots,s\}.
 \end{cases}
 \label{eq:star-leaders}
\end{equation}
Both columns in the last line are retained, so these are valid
minimum-weight leaders. Among the \(16\) leaders, coordinate \(8\)
occurs \(s+1\) times, its \(s\) partners twice each, and each
other coordinate once. Thus
\[
 B_s:=B(E)=\frac{(7-s)^2+36s+49(14-2s)}{64},
\]
and consequently
\begin{equation}
 B_s-k=\frac{s^2-12s+31}{64},
 \label{eq:star-B-minus-k}
\end{equation}
which is positive for \(s=1,2,3\).

\begin{proposition}[Failure for every \(k\geq8\)]
\label{prop:k8-k10-counterexamples}
\label{cor:failure-all-k-ge-8}
For every integer \(k\geq8\), there is a quantizer
\(f:\mathbb F_2^{k+4}\to\mathbb F_2^k\) such that
\[
 I(f(X);Y)>kC(p)
\]
for all \(p<1/2\) sufficiently close to \(1/2\).
\end{proposition}
\begin{proof}
For \(k=10,9,8\), respectively, \eqref{eq:star-B-minus-k} gives
\(B_s-k=5/16,11/64,1/16>0\); apply
\eqref{eq:general-high-noise}.
For \(k>8\), append \(k-8\) independent coordinates to the
eight-bit quantizer. Their information adds \((k-8)C(p)\),
preserving the strict excess.
\end{proof}

Table~\ref{tab:shortened-hamming} summarizes numerical evaluations
of \(\Delta_k(p)=I(Q(X);Y)-kC(p)\). The displayed interval
endpoints and maxima are numerical estimates; the preceding proof
does not rely on them.

\begin{table*}[t]
\caption{Shortened-Hamming quantizers: exact high-noise excesses
and numerically estimated positive intervals and maxima.}
\label{tab:shortened-hamming}
\centering
{\footnotesize\setlength{\tabcolsep}{7pt}\renewcommand{\arraystretch}{1.12}
\begin{tabular}{@{}ccccc@{}}
\toprule
\(k\) & code & \(B_s-k\) & numerical \(p\)-interval
& \(\max_p\Delta_k(p)/k\) \\ \midrule
10 & \([14,10,3]\) & \(5/16\)
& \((0.088658,1/2)\) & \(4.48440\!\times\!10^{-3}\) \\
9 & \([13,9,3]\) & \(11/64\)
& \((0.152276,1/2)\) & \(1.83320\!\times\!10^{-3}\) \\
8 & \([12,8,3]\) & \(1/16\)
& \((0.265563,1/2)\) & \(3.21308\!\times\!10^{-4}\) \\
\bottomrule
\end{tabular}}
\end{table*}

\subsection{Searches Below Eight Bits}
\label{subsec:below-eight-searches}

For \(k=7,6,5\), the same construction gives
\(B_s-k=-1/64,-1/16,-5/64\), respectively.
At \(k=7\), exhaustive enumeration of all \(1365\) Hamming shortenings
and \(451{,}920\) minimum-weight leader tables gave the maximum
\(B=447/64=6.984375\). Numerical checks of the evaluated candidates
across noise levels found no violation.

Additional searches covered \(1941\) puncturings of the \([15,7,5]\)
BCH code; all \(651\) seven-dimensional codes between
\(\mathrm{RM}(1,4)\) and \(\mathrm{RM}(2,4)\);
\(203\) sampled Golay shortenings;
\(345{,}000\) distinct-column parity-check searches with redundancies
\(5,\ldots,8\); and all \(255\) seven-dimensional quotients of the
eight-bit Hamming quantizer, up to output relabeling. We also tested
repeated columns and
nonminimum coset representatives, with tie optimization where applicable.
No violation was found at the tested crossover probabilities.

For \(128\) uniformly sampled codebooks of \(2^7\) words in
\(\mathbb F_2^{15}\), using fixed random nearest-neighbor tie assignments,
the mean \(B(Q)\) was \(6.379431\), with range
\([6.342219,6.423068]\). No sampled quantizer violated the benchmark
at the tested probabilities. Heuristic tie optimization on \(64\)
additional codebooks raised the mean coefficient to \(6.6954\),
with maximum \(6.71365\), still below the shortened-Hamming value.
These searches provide evidence, rather than an optimality proof.

\subsection{A Seven-Bit Threshold Conjecture}
\label{subsec:seven-bit-conjecture}

Appending independent coordinate outputs propagates any counterexample
to every larger output size. Together with the constructions and
searches above, this motivates the following conjecture.

\begin{conjecture}[Seven-bit threshold]
\label{conj:seven-bit-threshold}
For every \(2\leq k\leq7\), every \(n\geq k\), every
\(f:\mathbb F_2^n\to\mathbb F_2^k\), and every \(0\leq p\leq1/2\),
\begin{equation}
 I(f(X);Y)\leq kC(p).
 \label{eq:seven-bit-conjecture}
\end{equation}
Equality is attained by a coordinate projection.
\end{conjecture}

By the same padding argument, proving the conjecture for \(k=7\)
would settle the entire range. A first target is \(B(Q)\leq7\)
for every seven-bit quantizer, which controls the leading high-noise
coefficient. The scalar theorem cannot simply be iterated: conditioning
on earlier output bits generally destroys the uniform product input
distribution required for the conditional information bounds.

%% file: conclusion.tex
\section{Conclusion}
\label{sec:conclusion}

We proved that a coordinate is the most informative Boolean function of
a uniform binary vector observed through a BSC. The proof combines
monotone sorting, entropy flow, an entropy-constrained edge envelope, and
spectral information carried by pivotal sets. Its key global feature is
the growth implication: a positive advantage over a dictator would have
to increase with channel correlation and therefore cannot be reconciled
with the noiseless endpoint.

The vector-valued problem exhibits a different behavior. Shortened
Hamming quantizers violate the coordinate benchmark for every fixed
\(k\geq8\), whereas the structured and random searches, concentrated on
\(k=7\), found no counterexample below eight bits. This leaves a concrete
threshold problem. A
particularly useful first step would be to prove the sharp high-noise
bound \(B(Q)\leq7\) for every seven-bit quantizer and characterize its
equality cases. Possible approaches to a complete result include a
conditional form of the scalar theorem or a multiway entropy-flow
inequality, since the ordinary chain rule does not preserve
the uniform cube measure after conditioning.

%% file: ai-disclosure.tex
\section*{Generative AI Disclosure}

The proof of the single-bit result, including its detailed mathematical derivations, was generated by generative AI systems; the authors' role was high-level direction, feedback, and final review. The multibit extension was entirely human-directed: the authors formulated the problem, selected the code-based constructions and computational questions, and interpreted the evidence and conjectures. Generative AI systems assisted with developing and checking proofs, implementing and verifying computations, and refining the exposition. The authors reviewed the complete manuscript and accept full responsibility for its mathematical claims, citations, and final text.

%% file: scalar.tex
\section{Scalar Edge and Entropy Inequalities}\label{app:scalar}\label{ck-the-scalar-and-edge-lemmas}

This appendix proves the scalar comparisons behind
Lemma~\ref{lem:edges} and Proposition~\ref{prop:local}.
Centering and convexity give the common entropy-budget bounds, and
profile estimates control their angular correction.
Entropy contraction and a mean correction give the
influential-coordinate criterion.

\subsection{Convex Lower Bound for Centered Edges}\label{centered-edges-give-a-convex-lower-bound}

For edge values \(p,q\in(-1,1)\), define
\begin{align*}
 \edgewidth&=\frac{|p-q|}{2},& H_e&=\frac{h(p)+h(q)}2,\\
 a(p,q)&=\frac{(p-q)(g(p)-g(q))}{4},\\
 j(p,q)&=\frac{(\arcsin p-\arcsin q)^2}{4}.
\end{align*}
Write \(a=a(p,q)\), \(j=j(p,q)\), and, for \(k>0\), define
\[
 \theta_{\rm edge}(k)=\arcsin(\tanh k),\qquad
 \corr(F(k))=k-\frac{\theta_{\rm edge}(k)^2}{\tanh k}.
\]
We prove
\begin{align}
 a(p,q)&\ge \edgewidth F^{-1}(H_e/\edgewidth),\nonumber\\
 a(p,q)-j(p,q)&\ge \edgewidth\corr(H_e/\edgewidth).
 \tag{S2}\label{eq:scalar-S2}
\end{align}
For \(\edgewidth,H_e>0\), both right-hand sides are convex in
\((\edgewidth,H_e)\) and decrease with \(H_e\). At
\(\edgewidth=0\), with \(H_e>0\), they extend continuously by zero.

\subsubsection{Centering property}\label{why-centering-the-edge-helps}

Assume \(p>q\), and write \(p=\tanh(v+k)\), \(q=\tanh(v-k)\),
where \(k>0\). Then \(a=\edgewidth k\). With
\(M=\cosh(2v)\) and \(C_k=\cosh(2k)\),
\begin{align*}
 \edgewidth&=\frac{\sinh(2k)}{M+C_k},\\
 H_e&=\frac12\log[2(M+C_k)]
       -\frac{v\sinh(2v)+k\sinh(2k)}{M+C_k}.
\end{align*}
For \(v>0\), the derivative of \(H_e/\edgewidth\) has the sign of
\(\log[2(\cosh x+C_k)]-x\coth x\), where \(x=2v\). It is nonnegative:
\begin{align*}
 \log[2(\cosh x+C_k)]
 &\ge x+2\log(1+e^{-x})\\
 &\ge x+\frac2{e^x+1}\\
 &\ge x+\frac{2x}{e^{2x}-1}=x\coth x.
\end{align*}
Here we used \(\log(1+y)\ge y/(1+y)\) and \(e^x-1\ge x\).
Evenness therefore gives
\[
                         H_e/\edgewidth \ge F(k).               \tag{S3}\label{eq:scalar-S3}
\]

For \(b_{\rm edge}=\sinh k\) and \(x=b_{\rm edge}/\cosh v\),
the half-angle identity gives
\begin{align*}
 \arcsin p-\arcsin q&=2\arctan x,\\
 \frac{j}{\edgewidth}
 &=\frac{b_{\rm edge}}{\sqrt{1+b_{\rm edge}^2}}(1+x^2)
   \left(\frac{\arctan x}{x}\right)^2.
\end{align*}
The logarithmic derivative of
\((1+x^2)(\arctan x/x)^2\) is
\(2(x-\arctan x)/[x(1+x^2)\arctan x]>0\). Thus \(x\le b_{\rm edge}\) gives
\[
                         j/\edgewidth \le\theta_{\rm edge}(k)^2/\tanh k.    \tag{S4}\label{eq:scalar-S4}
\]

\subsubsection{Averaging by convexity}\label{why-the-bounds-can-be-averaged}

Put \(L_k=\log(2\cosh k)\). Then
\[
 F'(k)=-\frac{L_k}{\sinh^2 k},\qquad
 (F^{-1})'(F(k))=-\frac{\sinh^2 k}{L_k}.
\]
The positive ratio increases with \(k\), since \(2L_k>\tanh^2 k\).
Hence \(F^{-1}\) is decreasing and convex.

Similarly,
\begin{align*}
 \frac d{dk}\left(k-\frac{\theta_{\rm edge}(k)^2}{\tanh k}\right)
 &=\left(1-\frac{\theta_{\rm edge}(k)}{\sinh k}\right)^2,\\
 \corr'(F(k))&=-\frac{(\sinh k-\theta_{\rm edge}(k))^2}{L_k}.
\end{align*}
The latter positive ratio increases: its derivative has the sign of
\(2L_k\sinh k-(\sinh k-\theta_{\rm edge}(k))>0\), since \(2L_k>1\).
Thus \(\corr\) is decreasing and convex, and is nonnegative by its
zero limit as \(k\downarrow0\).

Combining these facts with \eqref{eq:scalar-S3}--\eqref{eq:scalar-S4}
proves \eqref{eq:scalar-S2}.
For \(\psi\in\{F^{-1},\corr\}\), the perspectives
\(\edgewidth\psi(H_e/\edgewidth)\) are convex and decrease with \(H_e\),
so Jensen applies to \eqref{eq:scalar-S2}.

Let \(u=T_\rho f\), with \(f\) nonconstant monotone Boolean and
\(0<\rho<1\). Each active coordinate has \(\edgewidth>0\),
\(\E\edgewidth=\rho a_i\), and \(\E H_e=\E h(u)\).
For any upper bound \(H_{\rm budget}\ge\E h(u)>0\), Jensen and
decrease in entropy give
\begin{align*}
 \E a&\ge\rho a_i v_i,\\
 \E(a-j)&\ge\rho a_i[v_i-c(v_i)],\\
 F(v_i)&=H_{\rm budget}/(\rho a_i).
\end{align*}
Summing proves Lemma~\ref{lem:edges}; inactive coordinates contribute zero.

\subsection{Monotonicity of the Entropy Profile}
\label{a-short-proof-that-vfv-decreases}

\begin{appendixlemma}[Profile monotonicity]
\label{lem:global-profile}
The function \(vF(v)\) strictly decreases for \(v>0\).
\end{appendixlemma}

The function \((1-r^2)g(r)-r h(r)\) vanishes at both endpoints of
\([0,1]\) and has second derivative \(-r/(1-r^2)<0\). Hence
\[
                   h(r)<\frac{(1-r^2)g(r)}r,\qquad 0<r<1.
\]
With \(r=\tanh v\), we have \(r-v(1-r^2)>0\), since
\(\sinh(2v)>2v\). Substitution gives
\begin{align*}
 \frac{d}{dv}[vF(v)]
 &=\frac{h(r)[r-v(1-r^2)]}{r^2}-\frac{v^2(1-r^2)}r\\
 &<\frac{v(1-r^2)}{r^3}(r-v)<0,
\end{align*}
since \(r<v\).

Choose an upper height \(0<\cutoff\le\logodds\). For \(0<v\le\cutoff\) and
\(a=F(\logodds)/F(v)\), monotonicity gives
\(a/v\le F(\logodds)/[\cutoff F(\cutoff)]\).
Also \(c(v)\le v\), since \(\corr\ge0\).
The LSI branch of Lemma~\ref{lem:deficit} now gives
\[
 \unitcost(v)\le\rho+p\Klsi\rho^2\frac{F(\logodds )}{\cutoff F(\cutoff)}
                   \qquad(0<v\le\cutoff).                 \tag{S5}\label{eq:scalar-S5}
\]

\input{profile-rates}

\input{elementary-inputs}

\subsection{The Influential-Coordinate Criterion}\label{a-large-singleton-coefficient-already-proves-ck}

For Boolean \(f\) and \(0<\rho<1\), fix a coordinate and write
\(a=|\widehat f(\{i\})|\), \(m=\E f\), \(\delta=1-\rho^2\).
Then \(|m|+a\le1\).

For a cube field \(v\) taking values in \([-1,1]\), define
\(\operatorname{Ent}_{\phi}(v)=\E\phi(v)-\phi(\E v)\).
The standard entropy contraction following from cube
LSI~\cite{Gross,ODonnell} is
\[
 \operatorname{Ent}_{\phi}(T_\rho v)
                 \le\rho^2\operatorname{Ent}_{\phi}(v).       \tag{S6a}\label{eq:scalar-S6a}
\]
Here \(\Ent(w)=\E[w\log w]-(\E w)\log(\E w)\), with \(0\log0=0\).
Apply \(\Ent(T_\rho w)\le\rho^2\Ent(w)\) to \(w=1\pm v\)
and use
\(2\operatorname{Ent}_{\phi}(v)=\Ent(1+v)+\Ent(1-v)\).

Noise the chosen bit first, then apply \eqref{eq:scalar-S6a} to the
remaining coordinates. The intermediate values are \(\pm\rho\) on
pivotal edges and \(\pm1\) otherwise; pivotal probability is at least
\(a\), and the section means are \(m\pm\rho a\). Thus
\[
 \E h(T_\rho f)\ge\rho^2a\,h(\rho)
                      +\frac{\delta}{2}[h(m+\rho a)+h(m-\rho a)]. \tag{S6}\label{eq:scalar-S6}
\]

The required mean correction is
\[
 \phi(m)+\frac{\delta}{2}[h(m+\rho a)+h(m-\rho a)]
                              \ge\delta h(\rho a).    \tag{S7}\label{eq:scalar-S7}
\]
The difference vanishes with its first derivative at \(m=0\).
Its second derivative has the sign of
\[
 [1-(m+a)^2][1-(m-a)^2]
          +(1-\rho^2)a^2(1-m^2-a^2)\ge0
\]
on \(|m|+a<1\). Convexity proves \eqref{eq:scalar-S7} there,
and continuity covers \(|m|+a=1\).

Combining \eqref{eq:scalar-S6}--\eqref{eq:scalar-S7} proves CK whenever \[
                 (1-\rho^2)h(\rho a)
                    -(1-\rho^2a)h(\rho)\ge0.          \tag{S8}\label{eq:scalar-S8}
\]
Section~\ref{app:local-propagation} extends a successful cutoff to
larger \(a\) and smaller correlations.

%% file: profile-rates.tex
\subsection{Relative Rates of the Entropy Profile and Angle Correction}
\label{app:profile-angle-rates}

\begin{appendixlemma}[The profile pays for the angle correction]
\label{lem:profile-angle-rates}
Define the logarithmic decay rate \(Q(v)=-F'(v)/F(v)\). Then
\begin{equation}
 \begin{aligned}
 Q(v)-\frac1v&>\frac32[1-c'(v)]&& (v>0),\\
 Q(v)-\frac1v&>2[1-c'(v)]&& (0<v\le2).
 \end{aligned}
 \label{eq:middle-profile-angle-rates}
\end{equation}
\end{appendixlemma}

\begin{proof}
Write
\(r_v=\tanh v\), \(\tau_v=\sech v\),
\(\theta_v=\arcsin r_v\), and \(s(v)=\theta_v/\sinh v\).
Then \(0<\theta_v<v<\sinh v\), and
\[
 s'=\frac{\tau_v^2-s}{r_v}<0,\qquad
 c'=2s-s^2,\qquad c''=2(1-s)s'<0.
\]
Here \(0<s<1\) and \(\theta_v>r_v\tau_v\). Thus \(c\) increases,
\(1-c'=(1-s)^2\), and \(c(v)/v\) decreases by concavity and \(c(0)=0\).

We also use
\begin{equation}
 \frac{3z}{2+\sqrt{1-z^2}}\le\arcsin z
 \le\frac{\pi z}{2+\sqrt{1-z^2}}
 \qquad(0\le z\le1).
 \label{eq:middle-angle-brackets}
\end{equation}
Set \(z=\sin t\). The quotient \(t(2+\cos t)/\sin t\)
increases from \(3\) to \(\pi\) on
\((0,\pi/2)\): the numerator of its derivative vanishes at zero
and has derivative \(2\sin t(t-\sin t)>0\).

With \(q_v=e^{-2v}\), the entropy formula and
\(\log(1+q_v)\le q_v\) give
\[
 Q(v)=\frac{4q_v[v+\log(1+q_v)]}
 {(1-q_v)[(1+q_v)\log(1+q_v)+2vq_v]}
 \ge\frac{4v}{2v+1}.
\]
Also \(Q(v)>\coth v\): for
\begin{align*}
 \mathcal H(v)&=\cosh v\,h(\tanh v),\\
 D(v)&=v\coth v-\log(2\cosh v),
\end{align*}
we have \(D'=(\tanh v-v)/\sinh^2v<0\),
\(\lim_{v\to\infty}D(v)=0\), and
\(\mathcal H'=-\sinh v\,D<0\). Now use \(F=\mathcal H/\sinh v\).

The power series show that \((\cosh v-1)/v^2\) and \(\sinh v/v\)
increase, and geometric tail bounds give
\begin{align*}
 \cosh1-1&<\frac{1/2}{1-1/12}<\frac35,\\
 \sinh1-1&<\frac{1/6}{1-1/20}<\frac15,\\
 \cosh2-1&<\frac2{1-1/3}=3.
\end{align*}
Also \(\cosh4<33\), using \(e^4<55\) from
Section~\ref{app:fixed-arithmetic}.
For \(0<v\le1\), \(\cosh v\le1+3v^2/5\) and
\eqref{eq:middle-angle-brackets} yield
\[
 1-s\le\frac{2v^2}{5+2v^2},\qquad
 Q(v)-1/v>\coth v-1/v>\frac{5v}{18}>\frac v4.
\]
The second inequality uses
\(v\cosh v-\sinh v=\int_0^v t\sinh t\,dt\ge v^3/3\) and
\(\sinh v<6v/5\). Finally
\((5+2v^2)^2-32v^3\ge25-12v^2>0\), so
\(v/4>2[2v^2/(5+2v^2)]^2\).

For \(1\le v\le2\), use \(\cosh v\le1+3v^2/4\), giving
\(1-s\le v^2/(2+v^2)\). Multiplying
\[
 \frac{4v}{2v+1}-\frac1v-2\frac{v^4}{(2+v^2)^2},
\]
by the positive factor \(v(2v+1)(2+v^2)^2\) gives
\[
 v^3(15v-4v^2-8)+4(3v+1)(v-1)>0.
\]
Here \(15v-4v^2-8\ge3\) on \([1,2]\). This proves the stronger bound.

For \(2\le v\le4\), the bound \(\cosh v\le1+2v^2\) gives
\(1-s\le4v^2/(3+4v^2)\). Multiplying
\[
 \frac{4v}{2v+1}-\frac1v-\frac32
                   \frac{16v^4}{(3+4v^2)^2},
\]
by the positive factor \(v(2v+1)(3+4v^2)^2\) gives
\[
 16v^4(v-2)(v-3/2)+16v^3(2v-3)+3(4v^2-6v-3)>0.
\]
Each term is nonnegative, and the last quadratic is increasing
from one. For \(v\ge4\), simply use
\(Q(v)-1/v\ge2-2/(2v+1)-1/v>3/2\) and \((1-s)^2<1\).
\end{proof}

%% file: elementary-inputs.tex
\subsection{Propagation of the Local Cutoff}
\label{app:local-propagation}

For \(0\le r,a\le1\), write
\[
 M_{\rm loc}(r,a)=(1-r^2)h(ar)-(1-ar^2)h(r).
\]
Yu's threshold-propagation lemma~\cite[Lemma~4.9]{Yu}, in sign means,
states that \(M_{\rm loc}(r_1,a)\ge0\), with \(0<r_1<1\), implies
\(M_{\rm loc}(r,a)\ge0\) for \(0\le r\le r_1\).
Since \(M_{\rm loc}(r,\cdot)\) is concave and
\(M_{\rm loc}(r,1)=0\), a check at \((r_1,a_0)\), with
\(a_0\in[0,1]\), therefore covers
all \(0\le r\le r_1\) and \(a_0\le a\le1\).

%% file: centered-support.tex
\section{An Entropy Tangent for the Angular Error}
\label{app:centered-support}

The angular comparison is the common starting point for both growth
arguments in Section~\ref{sec:pivotal-middle}. We prove its pointwise
form \eqref{eq:centered-support} by viewing the angular error as part of
an entropy variable and taking a supporting line at the dictator value.

To prove \eqref{eq:centered-support}, fix
\(5/4\le\lambda\le\pi/2\) and set
\[
 A_{\rm sc}(x)=x\arcsin x,\qquad
 \Psi(x)=\phi(x)+(\arcsin x-\lambda x)^2/2.
\]
For \(0<t<\pi/2\), put \(x=\sin t\), \(c=\cos t\), and \(g=\atanh x\).
Differentiation with respect to \(t\) gives
\[
 K_\lambda:=\ddot A_{\rm sc}\dot\Psi-
                \dot A_{\rm sc}\ddot\Psi
 =g(1+c^2)-x(t^2+2)+\lambda^2(x-c^3t).
\]
Since \(x-c^3t>0\), it suffices to prove \(K_{5/4}>0\).
The elementary bounds
\begin{align*}
 \arctan y&\le\frac{y(15+4y^2)}{15+9y^2},\\
 \atanh y&\ge\frac{y(15-4y^2)}{15-9y^2}\qquad(0\le y<1)
\end{align*}
follow by differentiation and equality at zero. Substituting
\(y=\tan(t/2)\) gives
\[
 \frac tx\le\frac{19+11c}{3(1+c)(4+c)},\qquad
 \frac gx\ge\frac{11+19c}{3(1+c)(1+4c)}.
\]
These imply
\[
 \frac{16K_{5/4}}x\ge
 \frac{(1-c)Q(c)}{9(1+c)(4+c)^2(1+4c)},
\]
where
\begin{align*}
 Q(c)={}&375(2-6c-7c^2)^2+164(1-c)^5\\
 &+740c(1-c)^4+c^2(1-c)^3\\
 &+25c^3(1-c)^2+5925c^4(1-c)+8625c^5.
\end{align*}
Thus \(K_\lambda>0\) for \(0<t<\pi/2\). Moreover,
\[
 \left(\frac{\dot\Psi}{\dot A_{\rm sc}}\right)'
       =-\frac{K_\lambda}{\dot A_{\rm sc}^{\,2}}<0.
\]
The ratio tends to \(\pi/2-\lambda\ge0\) as \(t\uparrow\pi/2\),
so \(\dot\Psi>0\) and \(d^2A_{\rm sc}/d\Psi^2>0\).

Now let \(r=\tanh\ell\), \(\ell\ge31/20\), and
\(\lambda=\arcsin(r)/r\). Since \(e^{31/10}>22\), we have
\(r^2>5/6\), and the alternating sine bound gives
\[
 \frac{\sin(5r/4)}r
 \le\frac54-\frac{125}{384}r^2+\frac{3125}{122880}r^4<1.
\]
Thus \(5/4<\lambda\le\pi/2\), the upper bound following from
the endpoint chord of \(\arcsin\). At \(x=r\),
\(\Psi(r)=\phi(r)\) and
\[
 \frac{dA_{\rm sc}}{d\Psi}(r)
   =\frac{\arcsin r+r/\sqrt{1-r^2}}{\atanh r}
   =\alpha.
\]
The supporting line is therefore
\(A_{\rm sc}(x)-A_{\rm sc}(r)\ge
\alpha[\phi(x)-\phi(r)+(\arcsin x-\lambda x)^2/2]\),
which is \eqref{eq:centered-support}. Evenness and continuity
cover negative \(x\) and the endpoints.

%% file: pivotal-channel-checks.tex
\section{Channel-Parameter Bounds in the Intermediate Regime}
\label{app:middle-channel}

The intermediate-correlation proof reduces to three inequalities that
depend only on the channel. We establish the profile comparison needed
for the log-sum step, the angular-energy reserve, and the bound that
absorbs the output bias. These are precisely the inputs used to close
Proposition~\ref{prop:middle-growth}.

\begin{appendixlemma}[The profile, energy, and mean reserves]
\label{lem:middle-channel-checks}
Let \(31/20\le\ell\le7/2\), with the channel parameters and
profiles defined in Sections~\ref{sec:edges}
and~\ref{sec:pivotal-middle}. Then
\begin{equation}
 \log\frac{vF(v)}{\ell F(\ell)}
 \ge\beta[H(\ell)-H(v)]\qquad(0<v\le\ell),
 \label{eq:middle-profile-check}
\end{equation}
\begin{equation}
 \theta^2-rH(\ell)-\frac\eta2\ge0,
 \label{eq:middle-energy-check}
\end{equation}
\begin{equation}
 \gamma-\frac{2\theta^2}{\sinh^2\ell}-\frac\lambda\alpha
 \ge r\ell.
 \label{eq:middle-mean-check}
\end{equation}
\end{appendixlemma}

\begin{proof}
By Lemma~\ref{lem:profile-angle-rates}, \(c\) increases, while
\(c(v)/v\) and \(s(v)=\arcsin(\tanh v)/\sinh v\) decrease.
Also \(\lambda\) increases, and \(\alpha>2\) increases because it
averages the increasing function \(\sech t+\cosh t\) on \([0,v]\).
Hence \(\gamma\) increases. We use \eqref{eq:middle-angle-brackets}
and the elementary bounds
\begin{equation}
 z-\frac{z^3}{3}\le\arctan z
 \le z-\frac{z^3}{3(1+z^2)}\qquad(z\ge0).
 \label{eq:middle-atan-brackets}
\end{equation}
Fixed constants below follow from
Section~\ref{app:fixed-arithmetic} and finite exponential sums.

\paragraph{Profile bound.}
The identity
\[
 \beta(\ell)=\frac{2}{c(\ell)}\left(1+\frac1{\gamma(\ell)}\right)
\]
shows that \(\beta\) decreases. At \(\ell=31/20\),
Appendix~\ref{app:centered-support} gives \(\lambda>5/4\).
Together with \(r>19/21\) and \(\alpha>2\), this yields
\(c(31/20)>475/336>7/5\), \(\gamma(31/20)>5/2\), and
\(\beta(31/20)<2\).

At \(\ell=2\), \(155/21<e^2<37/5\) gives
\(r<27/28\), \(\sinh2>18/5\), and, by
\eqref{eq:middle-atan-brackets},
\(\theta=\pi/2-2\arctan(e^{-2})>13/10\).
Thus \(c(2)>7/4\), \(\lambda(2)>4/3\),
\(\alpha(2)>49/20\), and \(\gamma(2)>13/4\). Consequently
\(\beta(2)<136/91<3/2\).
Lemma~\ref{lem:profile-angle-rates} now gives
\(Q(v)-1/v\ge\beta(\ell)[1-c'(v)]\): use its coefficient two
when \(\ell\le2\), and three halves when \(\ell\ge2\).
Integration proves \eqref{eq:middle-profile-check}.

\paragraph{Energy bound.}
Using \(rH(\ell)=r\ell-\theta^2\),
\eqref{eq:middle-energy-check} becomes
\[
 \frac{c(\ell)}{\ell}
       \left[\frac32+\frac1{2(\gamma+1)}\right]\ge1.
\]
Both factors decrease. At \(\ell=7/2\), \(32<e^{7/2}<35\) gives
\begin{align*}
 \theta&>\frac\pi2-\frac1{16},& \lambda&<\frac85,\\
 \alpha&<\frac{8/5+35/2}{7/2}=\frac{191}{35},& \gamma&<9.
\end{align*}
Since \(c>\theta^2\), the product exceeds
\[
 \frac{31}{20}\frac{(\pi/2-1/16)^2}{7/2}>1,
\]
using \(\pi>157/50\).

\paragraph{Mean bound.}
We prove \eqref{eq:middle-mean-check} for \(\ell\ge\ell_0=3/2\).
Put \(\tau=\sech\ell\) and \(s=s(\ell)\). The exponential bounds
and \eqref{eq:middle-atan-brackets} give
\begin{equation}
 \begin{split}
 \frac{19}{21}<\tanh\ell_0<\frac{77}{85},\qquad
 \frac{75}{32}<\cosh\ell_0<\frac52,\\
 \frac{19}{9}<\sinh\ell_0<\frac{77}{36},\qquad
 \frac98<\theta(\ell_0)<\frac87.
 \end{split}
 \label{eq:middle-mean-endpoints}
\end{equation}
Hence \(81/154<s(\ell_0)<72/133\), and the decrease of \(s\)
gives \(2s^2<3/5\) throughout this range. Define
\[
 Y(\ell)=c(\ell)+\theta\cosh\ell-r\ell^2-\frac65\ell.
\]
Then \(Y/\ell=\gamma-r\ell-6/5\), and, with
\(T=\theta\cosh\ell-r\), differentiation gives
\[
 Y''=T\left[1-\frac{2(1-s)}{\sinh^2\ell}\right]
                       +2\tau^2\ell(r\ell-2).
\]
Here \(r>9/10\), \(\cosh\ell>7/3\),
\(\sinh^2\ell>40/9\), and \(\lambda>6/5\) by
\eqref{eq:middle-angle-brackets}. Thus \(T>81/50\), and its
bracket exceeds \(11/20\). Also \(\ell\tau^2\) decreases from a
value below \(3/10\), while \(r\ell-2>-13/20\). Hence
\[
 Y''>\frac{81}{50}\frac{11}{20}-\frac{39}{100}>\frac12.
\]
The same endpoint bounds give \(Y(\ell_0)>3/16\) and
\(Y'(\ell_0)>-1/5\). Taylor's inequality, with
\(t=\ell-\ell_0\ge0\), yields
\[
 Y(\ell)>\frac3{16}-\frac t5+\frac{t^2}{4}
 =\frac14\left(t-\frac25\right)^2+\frac{59}{400}>0.
\]
Thus \(\gamma-r\ell>6/5\).

Finally, \eqref{eq:middle-angle-brackets} gives
\[
 \frac\lambda\alpha
 =\frac{\ell\theta}{r(\theta+\sinh\ell)}
 \le\frac{\pi\ell}{(\pi+1)r+2\sinh\ell}.
\]
For \(Z(\ell)=(\pi+1)r+2\sinh\ell-(5\pi/3)\ell\),
\eqref{eq:middle-mean-endpoints} and \(\pi<22/7\) give
\begin{align*}
 Z''&=2\sinh\ell-2(\pi+1)r\tau^2>4-\frac{522}{343}>0,\\
 Z'(\ell_0)&>\frac{16}{25}+\frac{14}{3}-\frac{110}{21}>0,\\
 Z(\ell_0)&>\frac{19}{21}+\frac{38}{9}
                 -\frac{67}{42}\frac{22}{7}>0.
\end{align*}
Thus \(Z>0\) and \(\lambda/\alpha<3/5\). Consequently
\[
 \gamma-r\ell-2s^2-\frac\lambda\alpha
 >\frac65-\frac35-\frac35=0.
\]
This proves \eqref{eq:middle-mean-check}.
\end{proof}

%% file: low.tex
\section{The Small-Correlation Regime}
\label{app:low}

We prove CK for \(0\le r=\rho\le r_1=\rho_*:=\tanh(31/20)\)
using the local/Fourier-gap approach~\cite{Sam,YuPhi,Yu}.
Proposition~\ref{prop:local} handles large singleton coefficients;
the remaining functions are controlled by a singleton-energy bound
and quadratic entropy minorants.
Entropy contraction then extends the bound to the lowest correlations.

Write \(W_1(v)=\sum_i\widehat v(\{i\})^2\),
\(a=\max_i|\widehat f(\{i\})|\), and set
\[
 a_0=r_0=\frac58,\qquad
 M=1-a_0+a_0^2=\left(\frac78\right)^2.
\]
The case \(r=0\) is immediate.

\input{fourier-cap}

\subsection{Mean-Sensitive Energy Bound}

The preceding estimate gives
\begin{equation}
 a\le a_0\quad\Longrightarrow\quad W_1(f)\le M.
 \label{eq:fourier-cap}
\end{equation}
Section~\ref{app:low-local-check} covers \(a\ge a_0\).
For \(a\le a_0\), form the balanced lift
\[
 G(1,x)=f(x),\qquad G(-1,x)=-f(-x).
\]
Its singleton coefficients are \(m=\E f\) and those of \(f\), so
\(|m|\le a_0\) implies
\begin{equation}
                         m^2+W_1(f)=W_1(G)\le M.
 \label{eq:lift}
\end{equation}
For \(|m|\ge a_0\), entropy contraction
\eqref{eq:scalar-S6a} gives, when \(r>0\),
\[
 I_f(r)\le r^2h(m)<r^2/2\le\phi(r),
\]
since \(h(m)\le L-a_0^2/2-a_0^4/12
<7/10-a_0^2/2-a_0^4/12<1/2\). Hence assume \(|m|\le a_0\).

\subsection{Tangent Quadratic Minorants}

For \(0<t<1\), define
\begin{align*}
 A_t&=\frac{\log(2/(1+t))}{(1-t)^2},\\
 z_t&=2t-\frac{g(t)}{A_t},\\
 Q_t(s)&=A_t(1-s)(1+s-z_t).
\end{align*}
Then \(Q_t(t)=h(t)\), \(Q_t'(t)=h'(t)\), and \(Q_t(1)=h(1)=0\).
Moreover,
\begin{equation}
                         Q_t(s)\le h(s)\qquad(0\le s\le1).
 \label{eq:entropy-tangent-family}
\end{equation}
Indeed, \(D_t=h-Q_t\) satisfies \(D_t(t)=D_t'(t)=D_t(1)=0\),
and \(D_t''(s)=2A_t-1/(1-s^2)\) strictly decreases.
Necessarily \(D_t''(t)>0\); otherwise \(D_t(1)<0\).
Thus \(D_t\) decreases to zero on \([0,t]\), then rises and falls
to zero on \([t,1]\).

Take \(t_0=2/5\), \(t_1=2/3\), \(A_j=A_{t_j}\), \(z_j=z_{t_j}\), and
\begin{align*}
 \beta&=\frac{1/A_0-1/A_1}{z_1-z_0},\\
 \alpha&=\frac1{A_0}+\beta z_0,\\
 c_{\rm ent}(r)&=\frac1{\alpha-\beta r^2}.
\end{align*}
Section~\ref{app:low-arithmetic} verifies \(z_0<r_0^2\le r_1^2<z_1\).
For \(r\in[r_0,r_1]\), write
\(r^2=\mu z_0+(1-\mu)z_1\), with \(0\le\mu\le1\).
Combining \eqref{eq:entropy-tangent-family} with weights
\(\mu c_{\rm ent}/A_0\) and \((1-\mu)c_{\rm ent}/A_1\), which sum
to one, yields
\begin{equation}
 h(s)\ge c_{\rm ent}(r)(1-s)(1+s-r^2).
 \label{eq:minorant}
\end{equation}

\subsection{From Entropy to Fourier Energy}

Writing \(W_k=\sum_{|S|=k}\widehat f(S)^2\), we have
\[
 \E|T_r f|\ge\E[fT_r f]=\sum_{k=0}^n r^kW_k.
\]
Apply \eqref{eq:minorant} to \(s=|T_r f|\).
Since \(h\) is even, Parseval and
\(r^{k+2}-r^{2k}\ge0\) for \(k\ge2\) give
\[
 \E h(T_r f)\ge
 c_{\rm ent}(r)\bigl[1-r^2-(1-r^2)m^2-r^2(1-r)W_1\bigr].
\]
Insert \eqref{eq:lift} and \(\phi(m)\ge m^2/2\):
\begin{align}
 \E h(T_r f)+\phi(m)\ge{}&
 c_{\rm ent}(r)(1-r)(1+r-Mr^2)\notag\\
 &+\bigl[1/2-c_{\rm ent}(r)(1-2r^2+r^3)\bigr]m^2.
 \label{eq:low-credit}
\end{align}
The mean coefficient has the sign of
\(\alpha-2+(4-\beta)r^2-2r^3\), which increases on \([r_0,1]\)
because \(\beta<1\). At \(r_0\), the bounds
\(\alpha>53/40\), \(\beta<17/20\) give a value exceeding
\(43/640\), so this term is nonnegative.

\subsection{Positivity of the Remaining Gap}

It remains to prove \(B_0(r)\ge0\) on \([r_0,r_1]\), where
\[
 B_0(r)=(1-r)(1+r-Mr^2)-(\alpha-\beta r^2)h(r).
\]
Differentiation gives
\begin{equation}
 B_0''''(r)=
 \frac{2[\alpha(1+3r^2)-\beta(r^4-3r^2+6)]}{(1-r^2)^3}.
 \label{eq:low-curvature-shape}
\end{equation}
Its numerator increases with \(r^2\). Together with
\(B_0'''(r_0)<0\), this implies that \(B_0'''\) changes sign at most
once, from negative to positive; hence \(B_0''\) has no interior maximum.
The bounds in Section~\ref{app:low-arithmetic} therefore give
\(B_0''<1/3\) on \([r_0,4/5]\) and \(B_0''<0\) on \([4/5,r_1]\).
Since \(B_0'(r_0)<-1/16\),
\[
 B_0'(r)<-\frac1{16}+\frac{4/5-5/8}{3}=-\frac1{240}
 \quad(r_0\le r\le4/5),
\]
and \(B_0'\) decreases thereafter. Thus
\(B_0(r)\ge B_0(r_1)>0\), proving CK on \([r_0,r_1]\)
by \eqref{eq:low-credit}.

For the remaining case \(a,|m|\le a_0\), the fixed checks also give
the starting slack
\[
 c_{\rm ent}(r_0)(1-r_0)(1+r_0-Mr_0^2)+r_0^2/2-L>0.
\]
Thus \(I_f(r_0)<r_0^2/2\). Applying entropy contraction to
\(T_r f=T_{r/r_0}(T_{r_0}f)\) covers the rest:
\[
 I_f(r)\le(r/r_0)^2I_f(r_0)\le r^2/2\le\phi(r)
                   \qquad(0\le r\le r_0).
\]

\subsection{Fixed Comparisons}
\label{app:low-arithmetic}

For \(1\le y\le2\), put \(z=(y-1)/(y+1)\) and
\begin{align*}
 \ell_-(y)&=2\sum_{j=0}^{11}\frac{z^{2j+1}}{2j+1},\\
 \ell_+(y)&=\ell_-(y)+\frac{2z^{25}}{25(1-z^2)}.
\end{align*}
The logarithm series gives
\[
 \ell_-(y)\le\log y\le\ell_+(y).
\]
For other positive rational arguments, extract a power of two.
Applied to the defining logarithms of \(A_j,z_j,\alpha,\beta\), these give
\[
 \alpha>\frac{53}{40},\qquad 0<\beta<\frac{17}{20},\qquad
 z_0<r_0^2,\quad r_1^2<z_1.
\]
For the endpoint checks, use
\begin{align*}
 r_1&=\frac{e^{31/10}-1}{e^{31/10}+1},&
 g(r_1)&=\frac{31}{20},\\
 h(r_1)&=\log(1+e^{-31/10})+\frac{31/10}{1+e^{31/10}}.
\end{align*}
Bound the exponential by Section~\ref{app:fixed-arithmetic}.
Substituting these rational enclosures and \(\ell_\pm\) into
\(B_0\), its first three derivatives, and \(M_{\rm loc}\), with
outward bounds throughout, certifies the following margins.
\begin{center}\small
\begin{tabular}{@{}p{.71\linewidth}r@{}}
\toprule
Fixed expression & Bound\\ \midrule
\(M_{\rm loc}(r_1,5/8)\) & \(>1/1400\)\\
\(B_0'(5/8)\) & \(<-17/250\)\\
\(B_0''(5/8)\) & \(<7/30\)\\
\(B_0'''(5/8)\) & \(<-1\)\\
\(B_0''(4/5)\) & \(<-1/9\)\\
\(B_0''(r_1)\) & \(<-1/10\)\\
\(B_0(r_1)\) & \(>1/4000\)\\
Starting contraction slack at \(r_0=5/8\) & \(>1/500\)\\
\bottomrule
\end{tabular}
\end{center}

%% file: fourier-cap.tex
\subsection{Singleton-Energy Bound}
\label{app:fourier-cap-proof}

For Boolean \(f\), put
\[
 a_i=\widehat f(\{i\}),\qquad a=\max_i|a_i|,\qquad
 W=\sum_i a_i^2=W_1(f).
\]
We prove
\begin{equation}
                 a\le\frac58
                 \quad\Longrightarrow\quad W\le 1-a_0+a_0^2=M.
 \label{eq:elementary-fourier-cap}
\end{equation}

Orient a coordinate with coefficient \(a\ge0\) and let
\(s=(f_++f_-)/2\), \(d=(f_+-f_-)/2\).
Since \(s^2+d^2=1\), \(d\in\{-1,0,1\}\), and \(\E d=a\), Parseval gives
\begin{equation}
 W=a^2+W_1(s)\le a^2+\E s^2\le a^2+1-a.
 \label{eq:elementary-singleton-section}
\end{equation}
For \(3/8\le a\le5/8\), the last expression is at most
its common endpoint value \(M\).

\subsubsection{A Sign-Sum Bound}

For independent uniform signs \(X_i\), let \(Z=\sum_i b_iX_i\). We show
\begin{equation}
 \sum_i b_i^2=1,\qquad \max_i b_i^2\le\frac15
 \quad\Longrightarrow\quad \E|Z|<\frac78.
 \label{eq:elementary-sign-sum}
\end{equation}
The Hermite interpolant to \(\sqrt y\) at \(9/16,4,9\), matching
values and first derivatives, is
\begin{equation}
 \begin{split}
 D_{\rm mom}p(y)={}&637696y^5-17928800y^4\\
                  &+191446315y^3-1008169920y^2\\
                  &+3660186285y+1260006516,
 \end{split}
 \label{eq:moment-majorant}
\end{equation}
where \(D_{\rm mom}=4042912500\).
The interpolation remainder gives, for \(y>0\) away from the contacts
and some \(\xi>0\),
\[
 \sqrt y-p(y)=\frac{(\sqrt{\,\cdot\,})^{(6)}(\xi)}{6!}
                (y-9/16)^2(y-4)^2(y-9)^2\le0,
\]
since \((\sqrt{\,\cdot\,})^{(6)}(\xi)=-945/(64\xi^{11/2})<0\).
Continuity includes the contacts and zero, so \(|Z|\le p(Z^2)\).

Write \(\sigma_j=\sum_i b_i^j\) and \(t=\sigma_4\).
The coefficient bound and Jensen's inequality with weights \(b_i^2\) give
\begin{align*}
 0\le t&\le\frac15,& \sigma_6&\le\frac t5,\\
 \sigma_{10}&\le\frac{\sigma_8}{5},& \sigma_8&\ge t^3.
\end{align*}
Expanding the moments in \(\E e^{xZ}=\prod_i\cosh(b_ix)\) yields
\begin{align*}
 D_{\rm mom}\E p(Z^2)={}&3487476486-214438410t\\
 &+1507452800t^2\\
 &+(1459014320-4285317120t)\sigma_6\\
 &-2928765440\sigma_8+5060755456\sigma_{10}.
\end{align*}
The \(\sigma_6\) coefficient is positive on \([0,1/5]\).
After substituting \(\sigma_{10}\le\sigma_8/5\), the
\(\sigma_8\) coefficient is \(-9583071744/5\).
Thus \(\sigma_6\le t/5\) and \(\sigma_8\ge t^3\) imply
\(D_{\rm mom}\E p(Z^2)\le T_{\rm mom}(t)\), where
\begin{align*}
 T_{\rm mom}(t)={}&3487476486+77364454t\\
 &+650389376t^2-\frac{9583071744}{5}t^3.
\end{align*}
The derivative \(T_{\rm mom}'\) is concave and positive at both endpoints:
\begin{align*}
 T_{\rm mom}'(0)&=77364454>0,\\
 T_{\rm mom}'(1/5)&=\frac{13440810318}{125}>0.
\end{align*}
Hence \(T_{\rm mom}(t)\le T_{\rm mom}(1/5)\), and
\[
 7D_{\rm mom}-8T_{\rm mom}(1/5)=\frac{119582002252}{625}>0.
\]
This proves \eqref{eq:elementary-sign-sum}.

If \(a\le3/8\) and \(W\ge M\), take \(b_i=a_i/\sqrt W\).
Then \(\sum_i b_i^2=1\), \(\max_i b_i^2\le9/49<1/5\), and
\[
 \sqrt W=\E\!\left[f\sum_i b_iX_i\right]
       \le\E\left|\sum_i b_iX_i\right|<\frac78.
\]
This contradicts \(W\ge M\). Together with
\eqref{eq:elementary-singleton-section}, this proves
\eqref{eq:elementary-fourier-cap}.

\subsection{The Local Cutoff}
\label{app:low-local-check}

For the local margin \(M_{\rm loc}\) of
Section~\ref{app:local-propagation}, Section~\ref{app:low-arithmetic} gives
\[
 M_{\rm loc}(r_1,5/8)
  =(1-r_1^2)h(5r_1/8)-(1-5r_1^2/8)h(r_1)>0.
\]
Concavity in \(a\), equality at \(a=1\), and the propagation result
in Section~\ref{app:local-propagation} extend this to
\(a\ge5/8\), \(r\le r_1\). Proposition~\ref{prop:local} then applies.

%% file: endpoints.tex
\section{Endpoint Reduction for the Height Parameter}
\label{app:endpoints}

The tail argument requires a scalar comparison at every active coordinate
height. The logarithmic-Sobolev estimate covers heights below \(\ell/2\);
above that point, the Parseval yield has no interior maximum. We reduce
the remaining comparison to the join and the local cutoff, which are
verified in Appendix~\ref{app:verification}.

\subsection{Shape Properties and Endpoint Reduction}

Recall \(c(v)=\arcsin^2(\tanh v)/\tanh v\) and
\(F(v)=h(\tanh v)/\tanh v\). Set
\begin{gather*}
 J=1/F,\qquad \Lambda(v)=\log(2\cosh v),\\
 {\mathcal H}(v)=\cosh v\,h(\tanh v).
\end{gather*}
The proof of Lemma~\ref{lem:profile-angle-rates} gives
\(c''<0\), \(0<c'<1\), and \({\mathcal H}'<0\) on \(v>0\).
Since \(J'=\Lambda/{\mathcal H}^2\), also \(J''>0\).

For \(\ell=\logodds\ge7/2\), \(r=\rho\), and
\(p=(\lambda+1/\alpha)^2\), the Parseval yield is
\[
 Y_P(v)=\frac{r c(v)+p[d_2/2+(d_1-d_2/2)F(\ell)/F(v)]}{v}.
\]
Section~\ref{app:original-yield} proves that it has no interior
maximum. The LSI comparison covers \(v\le z=\ell/2\).
Set
\[
 a_{\rm cut}=\frac1{1+(8/3)e^{-\ell}},\qquad
 F(V)=\frac{F(\ell)}{a_{\rm cut}}.
\]
The decrease of \(F\) gives a unique \(0<V<\ell\).
Section~\ref{app:manual-local} shows that a counterexample must have
\(a_i<a_{\rm cut}\), and hence \(\ell_i<V\).
Thus only
\(Y_P(z),Y_P(V)<B/\ell\) remain when \(V>z\);
if \(V\le z\), the LSI comparison suffices.

\subsection{Entropy-Profile Decay}

The cutoff proof uses
\begin{equation}
 Q_F(v):=-F'(v)/F(v)>8/5\qquad(v>0).
 \label{eq:analytic-profile-speed}
\end{equation}
Substituting \(q=e^{-2v}\) and clearing positive denominators
reduces this to
\[
 vq(1+4q)+(5q-2+2q^2)\log(1+q)>0.
\]
Only a negative logarithmic coefficient needs consideration.
Using \(\log(1+q)<q\), it suffices that
\(v-2+(4v+5)e^{-2v}>0\). This is immediate for \(v\ge2\).
On \([0,2]\), \((2-v)e^{2v}/(4v+5)\) is maximized at
\((3+\sqrt{65})/8\in(11/8,7/5)\), giving the upper bound
\(5e^{14/5}/84<1\) (Section~\ref{app:fixed-arithmetic}).

For this cutoff, integration gives
\[
 \frac85(\ell-V)<\int_V^\ell Q_F(v)\,dv
 =\log\bigl(1+(8/3)e^{-\ell}\bigr)<\frac83e^{-\ell}.
\]
Hence \(0<\ell-V<(5/3)e^{-\ell}\).

\input{original-yield}

%% file: original-yield.tex
\subsection{Endpoint Principle for the Original Yield}
\label{app:original-yield}

We prove the curvature bound
\begin{equation}
 \begin{split}
 &\frac{-c''(v)[vJ'(v)-J(v)]}{J''(v)}-[c(v)-vc'(v)]<\frac5{32},\\
 &\hspace{12em}v>0.
 \end{split}
 \label{eq:original-shape}
\end{equation}
It implies that \(Y_P(v)=[r c(v)+A_0+A_1J(v)]/v\), with
\(A_0,A_1>0\), has no interior maximum if \(A_0/r\ge5/32\):
its derivative has the sign of
\[
 A_1-\frac{A_0+r[c(v)-vc'(v)]}{vJ'(v)-J(v)},
\]
and the fraction strictly decreases by \eqref{eq:original-shape}.
For the Parseval yield, \(A_0=pd_2/2\) and
\(A_1=p(d_1-d_2/2)F(\ell)\); Appendix~\ref{app:canonical-prerequisites}
verifies \(A_0/r>3/16>5/32\) for \(\ell\ge7/2\).

\paragraph{Entropy curvature.}
Write \(r_v=\tanh v\), \(\delta_v=1-r_v^2\), and
\(D=vJ'-J>0\). Chebyshev's integral inequality gives
\[
 \frac{h(r_v)}{\delta_v}
 =\int_0^1\frac{ds}{(1+s)(1-r_v^2s^2)}
 \le L\int_0^1\frac{ds}{1-r_v^2s^2}=\frac{Lv}{r_v}.
\]
For \(0<v\le1\), this is at most \(L\coth1<1\), since
\(e^2>7\) and \(L<3/4\). Thus
\({\mathcal H}''=(h(r_v)-\delta_v)/\sqrt{\delta_v}<0\).
Consequently \({\mathcal H}^{-2}\) is increasing and convex, and
\(J'''=(\Lambda{\mathcal H}^{-2})''>0\). Hence
\begin{equation}
 D=\int_0^v tJ''(t)\,dt<\frac{v^2}{2}J''(v)
 \qquad(0<v\le1).
 \label{eq:original-profile-curvature}
\end{equation}
We extend this to
\begin{equation}
 \frac{D}{J''}<\frac v2\qquad(v>0).
 \label{eq:profile-ratio-linear}
\end{equation}
Only \(v\ge1\) remains. Put
\(q=e^{-2v}\), \(T=\log(1+q)/q\), and \(K=2v+(1+q)T\).
Then
\begin{gather*}
 J=\frac{1-q}{qK},\qquad K'=2T,\\
 qD=1-\frac{1-q+2T}{K}+\frac{(1-q^2)T^2}{K^2},\\
 0<T'=2[T-1/(1+q)]<q.
\end{gather*}
The last bound is the trapezoid estimate for \(1/(1+q)\).
Using \(q<1/7\), \(1-q/2\le T<1\), \(K>2\), and
\[
 qK<e^{-2v}(2v+3/2)\le(7/2)e^{-2}<1/2,
\]
and dropping the positive derivative of \((1-q^2)T^2\), we obtain
\begin{align*}
 (qD)'\ge\frac1{K^2}\biggl[&2T(1-q+2T)-K(2q+2T')\\
 &-\frac{4(1-q^2)T^3}{K}\biggr]>0:
\end{align*}
the three terms are respectively greater than \(4\), less than
\(2\), and less than \(2\).
Since \((qD)'=q(vJ''-2D)\), this proves
\eqref{eq:profile-ratio-linear}.

\paragraph{Angular curvature.}
Put \(w=\arcsin(\tanh v)/\sinh v\), \(\tau_v=\sech v\),
and \(k=-c''>0\). The angle bound
\eqref{eq:middle-angle-brackets} gives
\(w\ge3\tau_v/(2+\tau_v)>\tau_v^2\).
Using \(w'=(\tau_v^2-w)/r_v<0\), differentiation yields
\[
 vw''-w'\ge
 \frac{\tau_v(1-\tau_v)}{r_v^2(2+\tau_v)}
 \bigl[(3+\tau_v)r_v+v\{4-(1+\tau_v)^3\}\bigr]>0.
\]
Indeed, \(r_v\ge v\tau_v\) bounds the bracket below by
\(v(3-2\tau_v^2-\tau_v^3)>0\).
Thus \(f=1-w\) is increasing and concave in \(v^2\), with
\(f(0)=0\). For \(0<t\le v\),
\(f(t)\ge(t/v)^2f(v)\) and \(f'(t)\ge(t/v)f'(v)\).
Since \(k=2ff'\),
\begin{equation}
 c(v)-vc'(v)=\int_0^v tk(t)\,dt\ge\frac{v^2}{5}k(v).
 \label{eq:angle-curvature-scaling}
\end{equation}
Maximizing the quadratic in
\(k=2(1-w)(w-\delta_v)/r_v\) gives
\(k\le r_v^3/2<1/2\). Therefore
\[
 \frac{kD}{J''}-(c-vc')
 <k(v)\left(\frac v2-\frac{v^2}5\right)\le\frac5{32},
\]
which proves \eqref{eq:original-shape}.

%% file: verification.tex
\section{Canonical High-Correlation Comparisons}
\label{app:verification}\label{app:tail}

This appendix establishes the high-correlation case of
Lemma~\ref{lem:coverage} by verifying the baseline conditions
\eqref{eq:canonical-baseline-conditions} and the yield comparison
\eqref{eq:scalar-test} for \(\ell\ge7/2\).
After excluding large influences by the local criterion, we check small
heights and the two endpoints of the remaining height interval.
Appendix~\ref{app:endpoints} extends those endpoint checks to the
whole interval.

\subsection{Reduction to Three Comparisons}
\label{app:scalar-map}

Write \(r=\rho=\tanh\ell\), \(x=e^{-\ell}\), \(L=\log2\),
\(\theta=\arcsin r\), and \(P=\pi^2/4\). Recall
\begin{equation}
 \begin{split}
 \lambda&=\theta/r,\qquad
 \alpha=(\theta+\sinh\ell)/\ell,\\
 \gamma&=\lambda\alpha,\qquad d=\ell/2+5/4.
 \end{split}
 \label{eq:common-policy}
\end{equation}
As before, \(\weight(s)=\gamma s/(\gamma+s)\),
\(d_1=[\weight(d)-\weight(1)]r^2\),
\(d_2=[\weight(d)-\weight(2)]r^4\), and \(d_b=d_1-d_2/2\).
The retained-variance baseline is
\begin{align*}
 p&=(\lambda+1/\alpha)^2,& b&=\lambda^2(2+1/\gamma),\\
 C_d&=p\weight(d)-b,& B&=2\theta^2+C_d\phi(r)/L.
\end{align*}
The local criterion covers \(a\ge a_{\rm cut}=1/[1+(8/3)x]\).
A possible violation therefore has all active heights \(v<V\), where
\(F(V)=F(\ell)/a_{\rm cut}\).
Below we verify the baseline conditions, the LSI comparison through
\(z=\ell/2\), and the Parseval comparisons at \(z\) and \(V\).
The endpoint principle of
Appendix~\ref{app:endpoints} fills \([z,V]\) when \(V>z\);
otherwise the LSI comparison suffices.

\input{manual-local}
\input{canonical-tail-checks}

\subsection{Fixed Constants}
\label{app:fixed-arithmetic}

For \(0<y<14\), the fixed arithmetic uses
\begin{align*}
 E_{12}(y)&:=\sum_{j=0}^{12}\frac{y^j}{j!},\\
 E_{12}(y)&<e^y<E_{12}(y)+\frac{y^{13}}{13!(1-y/14)},
\end{align*}
and, for \(|y|<1\),
\begin{align*}
 \log\frac{1+y}{1-y}
 &=2\sum_{j=0}^{N-1}\frac{y^{2j+1}}{2j+1}+R_N(y),\\
 |R_N(y)|&\le\frac{2|y|^{2N+1}}{(2N+1)(1-y^2)}.
\end{align*}
Use \(y=(q-1)/(q+1)\) for \(\log q\); the low-correlation
arithmetic takes \(N=12\). Five terms at \(y=1/3\), and Machin's identity
\(\pi=16\arctan(1/5)-4\arctan(1/239)\), with four/five and
two/three alternating terms respectively, give
\[
                 3.14159<\pi<3.1416,\qquad .69314<\log2<.69315.
\]
In particular,
\(25/8<\pi<22/7\), \(9/4<P<5/2\), and
\(2/3<L<3/4\). The positive exponential sums give
\(e^3>20\), \(e^{7/2}>33\), and
\(3/5<e^{-1/2}<5/8\).
Square-root bounds follow by squaring positive rational endpoints.

%% file: manual-local.tex
\subsection{Local Cutoff}
\label{app:manual-local}\label{app:analytic-local}

For \(\ell\ge7/2\), put \(x=e^{-\ell}<1/32\) and
\begin{equation}
                 a_{\rm cut}=\frac{1}{1+(8/3)x}.
 \label{eq:rational-local-cutoff}
\end{equation}
Write \(h_b(t)=-t\log t-(1-t)\log(1-t)\) and set
\begin{align*}
 \kappa&=\frac4{3+8x},& \eta&=\frac{3x}{4(1+x^2)},\\
 t&=\frac{x^2}{1+x^2},& q&=\kappa x(1+\eta),\\
 A&=\frac1{2x^2}+\frac3{4x}+1+\frac{x^2}{2}.
\end{align*}
Thus \(t=(1-r)/2\), \(q=(1-r a_{\rm cut})/2\). The margin
\((1-r^2)h(r a_{\rm cut})-(1-a_{\rm cut}r^2)h(r)\), divided by the positive factor
\(4\kappa x^3/(1+x^2)^2\), is
\[
 \Lambda=\frac{h_b(q)}{\kappa x}-Ah_b(t).
\]
The reciprocal midpoint and trapezoid bounds
\[
 \frac{2s}{2+s}\le\log(1+s)\le s-\frac{s^2}{2(1+s)}
 \qquad(s\ge0)
\]
give \(h_b(q)\ge q\log(1/q)+q-q^2/(2-q)\) and
\((1+\eta)\log(1+\eta)\le\eta+\eta^2/2\). Hence
\begin{align*}
 \frac{h_b(q)}{\kappa x}\ge{}&
 (1+\eta)(\ell-\log\kappa)+1\\
 &-\frac{\eta^2}{2}-\frac{\kappa x(1+\eta)^2}{2-q}.
\end{align*}
Since \(\eta\le3x/4\), the last two errors satisfy
\begin{align*}
 \frac{\eta^2}{2}+\frac{\kappa x(1+\eta)^2}{2-q}
 &\le\frac9{32}x^2+\frac{x(4+3x)^2}{12(2+4x-x^2)}\\
 &\le\frac23x+\frac15x^2.
\end{align*}
Indeed, the rational term is at most \(2x/3-x^2/12\), because
\[
 (8-x)(2+4x-x^2)-(4+3x)^2=x(6-21x+x^2)>0
\]
for \(x<1/4\), and \(9/32-1/12<1/5\).
Using \(h_b(t)=\log(1+x^2)+2\ell x^2/(1+x^2)\), the upper
logarithm bound gives
\begin{align*}
 Ah_b(t)&\le\frac12+\frac34x+\frac34x^2
                 +\frac{2A\ell x^2}{1+x^2}.
\end{align*}
The discarded remainder is
\(x^3[2x(1+x^2)-3]/[8(1+x^2)]<0\).
Since \(1+\eta-2Ax^2/(1+x^2)=-x[3/(4(1+x^2))+x]\), subtraction gives
\begin{equation}
 \begin{split}
 \Lambda\ge{}&\frac12-(1+\eta)\log\kappa
 -\frac{17}{12}x-\frac{19}{20}x^2\\
 &-\ell x\left[\frac{3}{4(1+x^2)}+x\right].
 \end{split}
 \label{eq:manual-local-cancellation}
\end{equation}

The same logarithm bound gives \(\log(4/3)<7/24\), while
\(\ell e^{-\ell}\) decreases. Thus
\begin{align*}
 (1+\eta)\log\kappa
 &<\left(1+\frac3{128}\right)\frac7{24}<\frac13,\\
 \frac{17}{12}x+\frac{19}{20}x^2
 &<\frac{17}{384}+\frac{19}{20480}<\frac1{20},\\
 \ell x\left[\frac3{4(1+x^2)}+x\right]
 &<\frac7{64}\left(\frac34+\frac1{32}\right)<\frac3{32}.
\end{align*}
Therefore \(\Lambda>1/2-1/3-1/20-3/32=11/480>0\).
The local margin is concave in \(a\) and vanishes at \(a=1\),
so \eqref{eq:local} holds throughout \([a_{\rm cut},1]\).

%% file: canonical-tail-checks.tex
\subsection{Canonical Penalty and Preliminary Bounds}
\label{app:canonical-prerequisites}

The fixed bounds in Section~\ref{app:fixed-arithmetic} give
\(x<1/33\), \(r>99/100\), \(\theta>3/2\), and
\(\lambda^2<P<5/2\), using \(\arcsin r<\pi r/2\).
Since \(\theta+\sinh\ell>1/(2x)\),
\begin{equation}
 \gamma>\frac{3e^\ell}{4\ell}>\ell+\frac72=2d+1.
 \label{eq:canonical-penalty-lower}
\end{equation}
The second inequality follows because \(3e^\ell-4\ell^2-14\ell\)
and its first two derivatives are positive on this half-line.

The exact identities
\begin{equation}
 \begin{split}
 p\weight(1)-b&=-\lambda^2,\\
 p\weight(2)-b&=-\frac{\lambda^2}{\gamma+2},\\
 C_d&=\lambda^2\left[d-2-\frac{(d-1)^2}{\gamma+d}\right].
 \end{split}
 \label{eq:canonical-multipliers}
\end{equation}
give \(C_d>0\), since \(\gamma(d-2)>1\), and
\(C_d<(5/2)(d-2)\), \(b<(5/2)(2+1/7)<6\).
Since \(0<1-1/(2L)<1/3\), the mean condition follows from
\begin{align*}
 \gamma-b-C_d\left(1-\frac1{2L}\right)
 &>\ell+\frac72-6-\frac56\left(\frac\ell2-\frac34\right)\\
 &=\frac16+\frac7{12}\left(\ell-\frac72\right)>0.
\end{align*}
The endpoint-principle coefficient satisfies
\begin{align*}
 \frac{p d_2}{2r}
 &=\frac p2(d-2)\frac\gamma{\gamma+d}
                   \frac\gamma{\gamma+2}r^3\\
 &>\frac3{16}>\frac5{32},
\end{align*}
using \(p/2>1\), \(d-2\ge1\), both fractions \(>1/2\), and
\(r^3>3/4\). Also \(0<d_2<d_1\), hence \(d_b>0\).

\paragraph{Entropy estimate.}
Put \(E=h(r)-L(1-r^2)\ge0\), using \(\phi(r)\le Lr^2\).
For every \(\ell\ge3/2\),
\begin{equation}
                     E\le x^2(2\ell-5/3).
 \label{eq:shared-entropy-error}
\end{equation}
Indeed, the entropy formula, \(q=x^2\), and
\(\log(1+q)/q\le(2+q)/[2(1+q)]\) give
\[
 2\ell-\frac E q-\frac53
 \ge\frac{4(L-2/3)-11q/6+5q^2/6}{(1+q)^2}>0.
\]
Its numerator exceeds \(8/81-11/120>0\), since
\(L-2/3>2/81\) and \(q<1/20\).

\subsection{Small Heights: Paired Spectral Costs}
\label{app:manual-small}

Let \(z=\ell/2\). The line-angle LSI comparison is
\begin{align*}
 M_{\rm small}&=\frac B\ell-r
       -\frac{pK_{\rm LS}r^2F(\ell)}{zF(z)},\\
 K_{\rm LS}&=\frac{\gamma^2 e^{2d-3}}{2(\gamma+d)(\gamma+1)}.
\end{align*}
Set \(a_d=d-1\), \(k=e^{-1/2}\),
\(T=e^\ell F(\ell)/F(z)\), and \(f_d=C_d/\lambda^2\).
By \eqref{eq:canonical-multipliers} and \(\phi(r)/L=r^2-E/L\),
\begin{equation}
 \begin{split}
 \frac{\ell M_{\rm small}}{\theta^2}
 ={}&d-\frac\ell{c(\ell)}-kT\\
 &-\frac{a_d(a_d-kT)}{\gamma+d}-\frac{f_dE}{Lr^2}.
 \end{split}
 \label{eq:canonical-small-pairing}
\end{equation}
The bounds \(t/(1+t)\le\log(1+t)\) and
\((1+t)\log(1+t)\le t+t^2/2\) give
\begin{align*}
 \frac{2\ell+1}{(\ell+1+x/2)(1+x)}
 &\le T\le\frac{2\ell+1+x^2/2}{(\ell+1)(1+x)}\\
 &\le2-\frac1{\ell+1}.
\end{align*}
The lower bound exceeds \(5/3\), because
\[
 \ell-2-5\ell x-\frac{15}2x-\frac52x^2
 \ge\frac32-\frac{25}{32}-\frac5{2048}>0.
\]
Here the numerator increases with \(\ell\) and decreases with \(x\).
Using \(3/5<k<5/8\) and \(c(\ell)>\theta^2>9/4\),
\begin{align*}
 1&<kT<\frac58\left(2-\frac1{\ell+1}\right),\\
 d-\frac\ell{c(\ell)}-kT
 &>\frac\ell{18}+\frac5{8(\ell+1)}\ge\frac13,
\end{align*}
the last difference from \(1/3\) being
\((2\ell-3)(2\ell-7)/[72(\ell+1)]\).
Also \(\gamma+d>4(d-1)(d-2)\): with \(s=\ell-7/2\ge0\),
the cubic Taylor lower bound gives
\begin{gather*}
 3e^\ell-(4\ell^3-6\ell^2-8\ell)\\
 >29+2s+\frac{27}2s^2+\frac{25}2s^3>0.
\end{gather*}
Together with \eqref{eq:canonical-penalty-lower}, this yields
\[
 \frac{a_d(a_d-kT)}{\gamma+d}
 <\frac{a_d(a_d-1)}{\gamma+d}<\frac14.
\]
Finally, \eqref{eq:shared-entropy-error} gives
\[
 \frac{f_dE}{Lr^2}<2(d-2)(2\ell-5/3)e^{-2\ell}<\frac1{96},
\]
using \(f_d<d-2\), \(Lr^2>1/2\), and monotonicity of the
product, whose value at \(7/2\) is below \(2(16/3)/32^2=1/96\).
Thus
\[
 M_{\rm small}>\frac{\theta^2}{\ell}
                      \left(\frac13-\frac14-\frac1{96}\right)
                =\frac{7\theta^2}{96\ell}>0.
\]
Since \(vF(v)\) decreases and \(c(v)\le v\), this covers
\(0<v\le z\) and gives \(B/\ell>r\).

\subsection{Half-Height Join and Angular Saving}
\label{app:joining-height}\label{app:manual-join-direct}

Write \(a_z=F(\ell)/F(z)\). The original Parseval margin at the
join has the exact form
\begin{align*}
 M_{\rm join}
 &:={Bz}/{\ell}-r c(z)-p[d_2/2+d_ba_z]\\
 &=r[c(\ell)-c(z)]-\frac{\lambda^2r^4}{2(\gamma+2)}
   +\frac{C_d}{2}\left[1-r^4-\frac{h(r)}L\right]\\
 &\quad-p d_ba_z.
\end{align*}
The multiplier identities also give
\begin{equation}
 \begin{split}
 p d_b={}&C_d(r^2-r^4/2)
       +\lambda^2\left[r^2-\frac{r^4}{2(\gamma+2)}\right]\\
       &\le C_d/2+P\le Pd/2.
 \end{split}
 \label{eq:canonical-db-pairing}
\end{equation}

To prove that \(e^v[P-c(v)]\) increases, write
\(t=\arcsin(\tanh v)\). Then
\begin{align*}
 &(t+\cos t)^2-c(v)-c'(v)\\
 &\quad=(1-\sin t)(t/\sin t-\cos t)^2
 +\sin t\cos^2t>0,
\end{align*}
and \(t+\cos t<\pi/2\), so \(c+c'<P\).
At \(v_0=\log(10/3)\), \(4/\pi>5/4\) and
\(\arctan(3/10)\ge3/10-(3/10)^3/3\) give
\[
 2t/\pi<16/25,\qquad
 c(v_0)/P<(16/25)^2(109/91)<1/2.
\]
Since \(z\ge7/4>v_0\), \(P-c(z)>(5P/3)\sqrt x\).
Combining this with
\(c(\ell)\ge\theta^2\ge P-2\pi x\) gives
\[
                 c(\ell)-c(z)>\frac{5P}{3}\sqrt x-2\pi x.
\]
Using \(h(r)<(2\ell+1)x^2\), \(a_z=xT<2x\),
\eqref{eq:canonical-db-pairing}, and
\eqref{eq:canonical-penalty-lower}, and dropping
\(C_d(1-r^4)/2>0\), gives
\begin{align*}
 \frac{M_{\rm join}}{P\sqrt x}
 >{}&\frac{5r}{3}
 -\left(\frac76\ell+\frac54+\frac8\pi\right)e^{-\ell/2}\\
 &-\frac{(d-2)(2\ell+1)}{2L}e^{-3\ell/2}.
\end{align*}
The first term increases and both losses decrease. At \(7/2\),
\(r>99/100\), \(\sqrt x<7/40\), \(8/\pi<8/3\), and
\(2L>4/3\) bound the losses by \(8(7/40)\) and
\(6(7/40)^3<1/25\). Hence
\[
             \frac{M_{\rm join}}{P\sqrt x}
                  >\frac{33}{20}-\frac75-\frac1{25}
                   =\frac{21}{100}>0.
\]

\input{manual-cutoff-common}

%% file: manual-cutoff-common.tex
\subsection{Cutoff Bound with Exact Low-Degree Multipliers}
\label{app:manual-cutoff-common}\label{app:pivotal-tail-extension}

Let \(F(V)=F(\ell)/a_{\rm cut}\). By
\eqref{eq:analytic-profile-speed},
\[
                  0<\ell-V<\frac53x.
\]
The Parseval margin is
\[
 M_{\rm cut}=\frac B\ell V-r c(V)
                  -p[d_2/2+d_ba_{\rm cut}].
\]
The identity
\[
 B-pd_1-r c(\ell)=C_d[\phi(r)/L-r^2]=-C_dE/L
\]
and \(c'>0\) give
\[
 M_{\rm cut}\ge-C_dE/L+p d_b(1-a_{\rm cut})
                              -\frac B\ell(\ell-V).
\]
Using \(B\le2\lambda^2r^2+C_d\), the equality in
\eqref{eq:canonical-db-pairing}, and \eqref{eq:shared-entropy-error},
\begin{align*}
 \frac{M_{\rm cut}}x>{}&C_d\left[
 \frac{8r^2(1-r^2/2)}{3+8x}-\frac{(2\ell-5/3)x}{L}
                                      -\frac5{3\ell}\right]\\
 &+\lambda^2r^2\left[
 \frac8{3+8x}\left(1-\frac{r^2}{2(\gamma+2)}\right)
                                      -\frac{10}{3\ell}\right].
\end{align*}
Since \(r^2>3/4\), \(x<1/32\), and
\((2\ell-5/3)e^{-\ell}\) decreases,
\begin{align*}
 \frac{8r^2(1-r^2/2)}{3+8x}&>\frac{32}{13}\frac{15}{32}>1,\\
 \frac{(2\ell-5/3)x}{L}&<\frac{16}{3}\frac1{32}\frac32=\frac14,\\
 \frac5{3\ell}&\le\frac{10}{21}<\frac12.
\end{align*}
Thus the first bracket exceeds \(1/4\); the second exceeds
\(2(3/4)-1=1/2\), using \(\gamma>7\) and
\(10/(3\ell)\le20/21<1\). Therefore
\[
              M_{\rm cut}>x\left(\frac{C_d}{4}
                                  +\frac{\lambda^2r^2}{2}\right)>0.
\]
Together with the local cutoff and the endpoint principle, this
completes the height comparison for every \(\ell\ge7/2\).